\documentclass[
1p,
]{elsarticle}

\usepackage{verbatim}
\usepackage{float}
\usepackage{pgfplots}
	\pgfplotsset{compat=1.12}
\usepackage{tikz}
 	\usetikzlibrary{decorations.markings}
 	\usetikzlibrary{calc}
\usepackage{amsmath,amsthm,amssymb}

\usepackage{lmodern} 
\usepackage{silence}
\usepackage[bookmarksdepth=3]{hyperref}
	\hypersetup{pdfauthor=author}
\usepackage[capitalize]{cleveref}
	\crefname{equation}{}{}
	\crefformat{section}{\textsection#2#1#3}
	\crefformat{appendix}{\textsection#2#1#3}

\theoremstyle{plain}
	\newtheorem{mtheorem}{Theorem}
	
	\newtheorem{theorem}{Theorem}[section]
	\newtheorem{lemma}[theorem]{Lemma}
	\newtheorem{proposition}[theorem]{Proposition}
	\newtheorem{corollary}[theorem]{Corollary}

\theoremstyle{definition}

	\newtheorem{conjecture}[theorem]{Conjecture}
	\newtheorem{remark}[theorem]{Remark}
	\newtheorem{example}[theorem]{Example}

\newcommand{\N}{\mathbb{N}}
\newcommand{\Z}{\mathbb{Z}}
\newcommand{\R}{\mathbb{R}}
\newcommand{\C}{\mathbb{C}}
\newcommand{\T}{\mathbb{T}}

\newcommand{\rL}{\mathrm{L}}
\newcommand{\rR}{\mathrm{R}}

\newcommand{\cF}{\mathcal{F}}

\newcommand{\cH}{\mathcal{H}}

\DeclareMathOperator{\tr}{tr}

\newcommand{\sgn}{\mathrm{sgn}\,}

\newcommand{\wn}{\mathrm{wn}}
\renewcommand{\Re}{\mathrm{Re}\,}

\newcommand{\ind}{\mathrm{ind}\,}

\newcommand{\ess}{\sigma_{\mathrm{ess}}}

\newcommand{\Usuz}{U_{\textnormal{suz}}}

\newcommand{\textbi}[1]{\textit{\textbf{#1}}}
\usepackage{ytableau}
\newcommand{\bigidentity}{{\scalebox{3}{$1$}}}

\usepackage[table]{colortbl}
\newcommand{\bigdot}{\boldsymbol{\cdot}}
\usepackage{mathtools}
\newcommand{\svdots}{\raisebox{3pt}{$\scalebox{.75}{\vdots}$}}  
\newcommand{\sddots}{\raisebox{3pt}{$\scalebox{.75}{$\ddots$}$}}
\newcommand{\sdots}{\raisebox{3pt}{$\scalebox{.75}{$\dots$}$}}

\begin{document}

\begin{frontmatter}

\title
{
Index Theorems for One-dimensional Chirally Symmetric Quantum Walks with Asymptotically Periodic Parameters
}

\author[Shinshu1]{Yasumichi Matsuzawa}
	\ead{myasu@shinshu-u.ac.jp}
\author[Shinshu2]{Yohei Tanaka\corref{corresponding}}
	\ead{20hs602a@shinshu-u.ac.jp}
\author[Hachinohe]{Kazuyuki Wada}
	\ead{wada-g@hachinohe.kosen-ac.jp}

\cortext[corresponding]{Corresponding author}

\address[Shinshu1]{Department of Mathematics, Faculty of Education, Shinshu University, 6-Ro, Nishi-nagano,
Nagano 380-8544, Japan}
\address[Shinshu2]{Department of Science and Technology, Graduate School of Medicine, Science and Technology, Shinshu University, 4-17-1 Wakasato, Nagano, 380-8553, Japan}
\address[Hachinohe]{Department of General Science and Education, National Institute of Technology, Hachinohe College, Hachinohe 039-1192, Japan.}

\begin{abstract}
We focus on index theory for chirally symmetric discrete-time quantum walks on the one-dimensional integer lattice. Such a discrete-time quantum walk model can be characterised as a pair of a unitary self-adjoint operator $\varGamma$ and a unitary time-evolution operator $U,$ satisfying the chiral symmetry condition $U^* = \varGamma U \varGamma.$ The significance of this index theory lies in the fact that the index we assign to the pair $(\varGamma,U)$ gives a lower bound for the number of symmetry protected edge-states associated with the time-evolution $U.$  The symmetry protection of edge-states is one of the important features of the bulk-edge correspondence. The purpose of the present paper is to revisit the well-known bulk-edge correspondence for the split-step quantum walk on the one-dimensional integer lattice. The existing mathematics literature makes use of a fundamental assumption, known as the $2$-phase condition, but we completely replace it by the so-called asymptotically periodic assumption in this article. This generalisation heavily relies on analysis of some topological invariants associated with Toeplitz operators. 
\end{abstract}

\end{frontmatter}

%\newpage
%\setcounter{tocdepth}{3}
%\tableofcontents

% ------------------------------------------------------------------------------------------------------------ %
% New Section                                                                                                  %
% ------------------------------------------------------------------------------------------------------------ %
\section{Introduction}
The major mathematical theme of the present article can be broadly described as index theory for unitary operators. More specifically, we focus on an abstract unitary operator $U$ on a Hilbert space $\cH,$ satisfying the following \textbi{chiral symmetry condition}; 
\begin{equation}
\label{equation: chiral symmetry}
U^* = \varGamma U \varGamma,
\end{equation}
where $\varGamma$ is a fixed unitary self-adjoint operator on $\cH.$ %As is well-known, 
With the canonical decomposition $\cH = \ker(\varGamma - 1) \oplus \ker(\varGamma + 1)$ of the underlying Hilbert space in mind, we can naturally assign certain well-defined indices $\ind_+(\varGamma, U), \ind_-(\varGamma, U)$ to the given pair $(\varGamma, U),$ in such a way that the following estimate holds true;
\begin{equation}
\label{equation: topological protection of bounded states}
|\ind_\pm(\varGamma,U)| \leq \dim \ker(U \mp 1),
\end{equation}
where we assume that the essential spectrum of $U,$ denoted by $\ess(U),$ does not contain $\pm 1$ (see \cref{section: preliminaries} for the definition of $\ess(U)$). We typically impose the assumption of $\pm 1 \notin \ess(U)$ to ensure that the two indices on the left hand side of \cref{equation: topological protection of bounded states} can be viewed as well-defined Fredholm indices. As can be easily seen from \cref{equation: topological protection of bounded states}, if $\ind_\pm(\varGamma,U)$ is non-zero, then the eigenspace $\ker(U \mp 1)$ contains some non-trivial eigenstates. This implication, known as the \textit{protection of eigenstates by chiral symmetry}, is one of the important features of the bulk-edge correspondence. A brief summary of the index theory mentioned so far can be found in \cref{section: statement of the main theorems}. 

The symmetry protection of eigenstates naturally arises in the context of so-called (discrete-time) \textbi{quantum walks}  \cite{Gudder-1988,Aharonov-Davidovich-Zagury-1993,Meyer-1996,Ambainis-Bach-Nayak-Vishwanath-Watrous-2001}. Quantum walk models are typically characterised by their associated unitary time-evolution operators $U,$ and we may consider various symmetry types in a sense analogous to \cref{equation: chiral symmetry}. Index theory for quantum walks on the integer lattice $\Z$ with various symmetry types has been a particularly active subject of recent mathematical studies of quantum walks \cite{Cedzich-Grunbaum-Stahl-Velazquez-Werner-Werner-2016,Cedzich-Geib-Grunbaum-Stahl-Velazquez-Werner-Werner-2018,Cedzich-Geib-Stahl-Velazquez-Werner-Werner-2018}. In particular, the chiral symmetry condition \cref{equation: chiral symmetry} alone has attracted tremendous attention \cite{Suzuki-2019,Suzuki-Tanaka-2019,Matsuzawa-2020,Tanaka-2020,Cedzich-Geib-Werner-Werner-2021}, and it is also the main subject of the present article. \textbi{Suzuki's split-step quantum walk} \cite{Fuda-Funakawa-Suzuki-2017,Fuda-Funakawa-Suzuki-2018,Fuda-Funakawa-Suzuki-2019,Tanaka-2020,Narimatsu-Ohno-Wada-2021} can be viewed as a prominent example of a one-dimensional chirally symmetric quantum walk. This model is characterised by the following two operators defined on the Hilbert space  $\ell^2(\Z, \C^2) = \ell^2(\Z) \oplus \ell^2(\Z)$ of $\C^2$-valued square-summable sequences:
\begin{align}
\label{equation: suzuki split-step quantum walk}
\varGamma_{\textnormal{suz}} 
&:= 
\begin{pmatrix}
1 & 0 \\
0  &  L^*
\end{pmatrix}
\begin{pmatrix}
p & \sqrt{1 - p^2} \\
\sqrt{1 - p^2} & -p
\end{pmatrix}
\begin{pmatrix}
1 & 0 \\
0  &  L
\end{pmatrix}, \\
% ---------- % 
\label{equation: definition of Suzuki evolution operator}
U_{\textnormal{suz}} 
&:= 
\begin{pmatrix}
1 & 0 \\
0  &  L^*
\end{pmatrix}
\begin{pmatrix}
p & \sqrt{1 - p^2} \\
\sqrt{1 - p^2} & -p
\end{pmatrix}
\begin{pmatrix}
1 & 0 \\
0  &  L
\end{pmatrix}
\begin{pmatrix}
a & b^* \\
b & -a
\end{pmatrix},
\end{align}
where $L$ is the bilateral left-shift operator on $\ell^2(\Z)$ (see \cref{definition: bileteral shift} for definition), and where $p = (p(x))_{x \in \Z}$ and $a = (a(x))_{x \in \Z}$ are real-valued sequences assuming values in the closed interval $[-1,1].$ Such bounded sequences are identified with the corresponding bounded multiplication operators on $\ell^2(\Z)$ throughout this paper. It is shown in \cite[Corollary 4.4]{Narimatsu-Ohno-Wada-2021} that the model introduced above is equivalent to Kitagawa's split-step quantum walk \cite{Kitagawa-Rudner-Berg-Demler-2010, Kitagawa-Broome-Fedrizzi-Rudner-Berg-Kassal-Aspuru-Demler-White-2012, Kitagawa-2012}.

It follows from a direct computation that if we set $(\varGamma, U) := (\varGamma_{\textnormal{suz}} ,U_{\textnormal{suz}}),$ then we obtain the chiral symmetry condition \cref{equation: chiral symmetry}. To give a complete classification of the associated indices $\ind_\pm(\varGamma, U),$ let us assume the existence of the following limits for each $\star = -\infty, +\infty:$
\begin{align}
\label{equation: anisotropic assumption}
&p(\star) := \lim_{x \to \star} p(x), & &a(\star) := \lim_{x \to \star} a(x).
\end{align}  
It is shown in \cite[Theorem 1.1]{Matsuzawa-Seki-Tanaka-2021}(i) that under \cref{equation: anisotropic assumption}, we have $\pm 1 \notin \ess(U)$ if and only if $p(\star) \neq \pm a(\star)$ for each $\star = -\infty, +\infty.$ Moreover, in this case
\begin{equation}
\label{equation2: GW indices for ssqw with periodic parameters}
\ind_\pm(\varGamma, U) = 
\begin{cases}
+1, & p(-\infty) \mp a(-\infty) < 0 < p(+\infty) \mp a(+\infty), \\
-1, & p(+\infty) \mp a(+\infty) < 0 < p(-\infty) \mp a(-\infty), \\
0, & \mbox{otherwise}.
\end{cases}
\end{equation}
Note first that the index formula \cref{equation2: GW indices for ssqw with periodic parameters} is robust in the sense it depends only on the asymptotic values \cref{equation: anisotropic assumption}. In particular, if $|\ind_\pm(\varGamma,U)| = 1,$ then it follows from \cref{equation: topological protection of bounded states} that the eigenspace $\ker(U \mp 1)$ contains at least one non-trivial eigenstate. It is also shown in \cite[Theorem 1.1]{Matsuzawa-Seki-Tanaka-2021}(ii) that such symmetry protected eigenstates exhibit exponential decay in a certain well-defined sense.

This begs the following natural question. The existence of the two-sided limits \cref{equation: anisotropic assumption} is a fundamental assumption in \cite[Theorem 1.1]{Matsuzawa-Seki-Tanaka-2021}, but is there some meaningful way to generalise this result? The ultimate purpose of the present article is to show that such a generalisation is actually possible, and we do so by replacing \cref{equation: anisotropic assumption} with the so-called \textbi{asymptotically periodic assumption} (see \cref{section: strictly local operators with asymptotically periodic parameters} for definition).

The present article is organised as follows. The evolution operator of Suzuki's split-step quantum walk given by \cref{equation: definition of Suzuki evolution operator} is an explicit example of so-called \textbi{strictly local operators.} In \cref{section: strictly local operators with asymptotically periodic parameters} we develop an elementary operator-algebraic method to classify some topological invariants associated with strictly local operators satisfying the asymptotically periodic assumption. In fact, this result is a gneralisation of \cite[Theorem A]{Tanaka-2020}. In \cref{section: applications of Theorem A} we give a generalisation of \cite[Theorem 1.1]{Matsuzawa-Seki-Tanaka-2021} as a direct application of \cref{section: strictly local operators with asymptotically periodic parameters}. The paper concludes with several concluding remarks in \cref{section: discussion}.

On a final note, the present article focuses on some topological invariants associated with quantum walks, which only make sense in infinite dimensions. The novelty of our approach lies in the fact we can fully classify these topological invariants in the language of linear algebra. As we shall see in this paper, some crucially important arguments can be eventually simplified to analysis of $n \times n$ matrices of the form;
\begin{equation}
\label{equation: almost tridiagonal matrix}
\begin{pmatrix}
\alpha_{0} & \beta_0 & 0 & \cdots & 0 &  \gamma_0 \\
\gamma_{1} & \alpha_{1} & \beta_1 & \cdots & 0 & 0 \\
0 & \gamma_{2} & \alpha_{2} & \cdots & 0 & 0 \\
\svdots & \svdots & \svdots & \sddots & \svdots & \svdots \\
0 & 0 & 0 & \cdots & \alpha_{n-2} & \beta_{n-2} \\
\beta_{n-1}  & 0 & 0 & \cdots & \gamma_{n-1}& \alpha_{n-1} \\
\end{pmatrix}.
\end{equation}

% ------------------------------------------------------------------------------------------------------------ %
% New Section                                                                                                  %
% ------------------------------------------------------------------------------------------------------------ %
\section{Strictly local operators with asymptotically periodic parameters}
\label{section: strictly local operators with asymptotically periodic parameters}

We start with the following main result of \cite{Tanaka-2020};
\begin{theorem}[{\cite[Theorem A]{Tanaka-2020}}]
\label{theorem: Tanaka2020}
%\label{theorem: topological invariants of strictly local operators with asymptotically periodic parameters}
Let $k_0 \in \N,$ and let $A_{-k_0}, \dots, A_{k_0}$ be $n \times n$ matrices-valued sequences on $\Z$ admitting the following limits for $-k_0 \leq k \leq k_0;$
\begin{equation}
\label{equation: two-side limits of Ay}
A_{k}(\rL) := \lim_{x \to -\infty} A_k(x), \qquad 
A_{k}(\rR) := \lim_{x \to +\infty} A_k(x).
\end{equation}
Let 
\begin{align}
\label{equation: strictly local operator in simple form}
A 
&:= 
\sum^{k_0}_{k = -k_0} 
A_k
\begin{pmatrix}
L^{k} & \dots & 0 \\
\vdots & \ddots& \vdots \\
0 & \dots & L^{k} \\
\end{pmatrix}, \\
% ------------------- %
\label{equation: Ahat in simple form}
\hat{A}(\sharp, z) &:=
\sum^{k_0}_{k = -k_0} 
A_k(\sharp)
\begin{pmatrix}
z^{k} & \dots & 0 \\
\vdots & \ddots& \vdots \\
0 & \dots & z^{k} \\
\end{pmatrix}, & &z \in \T &&\sharp = \rL, \rR.
\end{align}
where $L$ is the bilateral left-shift operator on $\ell^2(\Z),$ and where each $A_k$ in \cref{equation: strictly local operator in simple form} is viewed as the bounded multiplication operator on $\ell^2(\Z, \C^n).$ Then the following assertions hold true:
\begin{enumerate}[(i)]
\item We have that $A$ is Fredholm if and only if $\T \ni z \longmapsto \det \hat{A}(\sharp, z) \in \C$ is nowhere vanishing on $\T$ for each $\sharp = \rL, \rR.$ In this case, the Fredholm index of $A$ is given by
\begin{equation}
\label{equation: index and winding number}
\ind(A) = \wn \left(\det \hat{A}(\rR,\cdot) \right) - \wn \left(\det \hat{A}(\rL,\cdot) \right),
\end{equation}
where $\wn \left(\det \hat{A}(\sharp, \cdot) \right)$ denotes the winding number of the continuous function $\T \ni z \longmapsto \det \hat{A}(\sharp, z) \in \C$ with respect to the origin for each $\sharp = \rL, \rR.$
\item The essential spectrum of $A$ is given by
\begin{align}
\label{equation: essential spectrum of A}
\ess(A) =  \bigcup_{z \in \T}  \sigma \left(\hat{A}(\rR, z) \right) \cup \bigcup_{z \in \T}  \sigma \left(\hat{A}(\rL,z) \right).
\end{align}
\end{enumerate}
\end{theorem}

Any operator $A$ of the form \cref{equation: strictly local operator in simple form} is referred to as an $n$-dimensional \textbi{strictly local operator} on the integer lattice $\Z$ throughout this paper. The purpose of the current section is to generalise the existing formulas \crefrange{equation: index and winding number}{equation: essential spectrum of A}, by replacing the assumption \cref{equation: two-side limits of Ay} with the so-called \textbi{asymptotically periodic assumption.} More precisely, we assume that there exist natural numbers $n_\rL, n_\rR$ with the property that the following limits exist for $-k_0 \leq k \leq k_0;$
\begin{align}
\label{equation: asymptotically periodic}
&A_k(\rL, m) := \lim_{x \to -\infty} A_k(n_\rL \cdot x + m), \qquad m \in \{0, \dots, n_\rL - 1\}, \\
&A_k(\rR, m) := \lim_{x \to +\infty} A_k(n_\rR \cdot x + m), \qquad m \in \{0, \dots, n_\rR - 1\}.
\end{align}
In other words, %%\cref{equation: asymptotically periodic} means that 
the doubly-infinite sequences $A_{-k}, \dots, A_{k}$ are \textbi{asymptotically $(n_\rL, n_\rR)$-periodic}. Let us consider the following explicit example;

\begin{comment}
More precisely, we assume that there exist natural numbers $n_\rL, n_\rR$ with the property that the following limits exist for each $\star = -\infty, +\infty;$
\begin{align}
\label{equation: asymptotically periodic}
&A_k(\star, n_0) := \lim_{x \to \star} A_k(n_\sharp x + n_0), & &n_0 \in \{0, \dots, n_\sharp -1\}, & &k \in \{-k_0, \dots, k_0\}.
\end{align}
In other words, \cref{equation: asymptotically periodic} means that the doubly-infinite sequences $A_{-k}, \dots, A_{k}$ are \textbi{asymptotically $(n_\rL, n_\rR)$-periodic}. Let us consider the following explicit example;
\end{comment}

\begin{remark}
Let $A_k = (A_k(x))_{x \in \Z}$ be asymptotically $(3, 2)$-periodic. That is, we assume the existence of the following $3 + 2 = 5$ limits:
\begin{align*}
A_k(\rL, 0) &= \lim_{x \to -\infty} A_k(3x + 0), & A_k(\rR, 0) &= \lim_{x \to +\infty} A_k(2x + 0), \\
A_k(\rL, 1) &= \lim_{x \to -\infty} A_k(3x + 1), & A_k(\rR, 1) &= \lim_{x \to +\infty} A_k(2x + 1), \\
A_k(\rL, 2) &= \lim_{x \to -\infty} A_k(3x + 2), & &
\end{align*}
On the other hand, one can rearrange $A_k(x)$ according to the following table; 
\begin{table}[H]
\begin{center}
{\tiny
\begin{tabular}{|c|c|c|c|c|c|c|c|c|c|} 
\hline
% $x = -\infty$ & $\dots$ & $x = -3$ & $x = -2$ & $x = -1$ & $x = 0$ & $x = 1$ & $x = 2$ & $\dots$ & $x = \rR$\\ \hline 
\cellcolor{gray!5}  $A_k(\rL, 0)$ & \cellcolor{gray!5} $\leftarrow$ & \cellcolor{gray!5}  $A_k(-9)$ & \cellcolor{gray!5}  $A_k(-6)$ & \cellcolor{gray!5}  $A_k(-3)$ & \cellcolor{black!5} & \cellcolor{black!5} & \cellcolor{black!5} & \cellcolor{black!5} $\dots$ & \cellcolor{black!5} \\ \hline 
\cellcolor{gray!5}  $A_k(\rL, 1)$ & \cellcolor{gray!5} $\leftarrow$ & \cellcolor{gray!5}  $A_k(-8)$ & \cellcolor{gray!5}  $A_k(-5)$ & \cellcolor{gray!5}  $A_k(-2)$ & \cellcolor{black!5} & \cellcolor{black!5} & \cellcolor{black!5} & \cellcolor{black!5} $\dots$ & \cellcolor{black!5} \\ \hline 
\cellcolor{gray!5}  $A_k(\rL, 2)$ & \cellcolor{gray!5} $\leftarrow$ & \cellcolor{gray!5}  $A_k(-7)$ & \cellcolor{gray!5}  $A_k(-4)$ & \cellcolor{gray!5}  $A_k(-1)$ & \cellcolor{black!5} & \cellcolor{black!5}& \cellcolor{black!5}  & \cellcolor{black!5} $\dots$ & \cellcolor{black!5} \\ \hline 
\cellcolor{gray!5}   & \cellcolor{gray!5} $\dots$ & \cellcolor{gray!5}   & \cellcolor{gray!5}  & \cellcolor{gray!5}   & \cellcolor{black!5} $A_k(0)$& \cellcolor{black!5} $A_k(2)$ & \cellcolor{black!5} $A_k(4)$& \cellcolor{black!5} $\rightarrow$ & \cellcolor{black!5} $A_k(\rR, 0)$ \\ \hline 
\cellcolor{gray!5}   & \cellcolor{gray!5} $\dots$ & \cellcolor{gray!5}   & \cellcolor{gray!5}  & \cellcolor{gray!5}   & \cellcolor{black!5} $A_k(1)$& \cellcolor{black!5} $A_k(3)$ & \cellcolor{black!5} $A_k(5)$& \cellcolor{black!5} $\rightarrow$ & \cellcolor{black!5} $A_k(\rR, 1)$ \\  \hline
\end{tabular}
}
\end{center}
\end{table}
The first three rows show that $A_k(3x + 0), A_k(3x + 1), A_k(3x + 2)$ have well-defined limits as $x \to -\infty,$ whereas the last two rows show that $A_k(2x + 0), A_k(2x + 1)$ have well-defined limits as $x \to +\infty.$
\end{remark}

We are now in a position to state the following generalisation of \cref{theorem: Tanaka2020};
\begin{mtheorem}
\label{theorem: topological invariants of strictly local operators with asymptotically periodic parameters}
Let $k_0 \in \N,$ and let $A_{-k_0}, \dots, A_{k_0}$ be finitely many $n \times n$ matrices-valued sequences on $\Z$ admitting the following representations:
\begin{equation}
A_k(x) = 
\begin{pmatrix}
a_{11}^k(x) & \dots & a_{1n}^k(x)  \\
\vdots & \ddots& \vdots \\
a_{n1}^k(x) & \dots & a_{nn}^k(x)  \\
\end{pmatrix}, \qquad x \in \Z, \qquad -k_0 \leq k \leq k_0.
\end{equation}
We assume that there exist $n_\rL, n_\rR \in \N$ with the property that the following limits exist for $1 \leq i,j \leq n$ and for $-k_0 \leq k \leq k_0;$
\begin{align}
\label{equation: asymptotical left limit}
&a_{ij}^k(\rL, m) := \lim_{x \to -\infty} a_{ij}^k(n_\rL \cdot x + m), \qquad m \in \{0, \dots, n_\rL - 1\}, \\
\label{equation: asymptotical right limit}
&a_{ij}^k(\rR, m) := \lim_{x \to +\infty} a_{ij}^k(n_\rR \cdot x + m), \qquad m \in \{0, \dots, n_\rR - 1\}.
\end{align}
\begin{comment}
for each $\star = -\infty, +\infty:$
\begin{align}
&a_{ij}^k(\star, n_0) := \lim_{x \to \star} a_{ij}^k(n_\sharp x + n_0), & 
&n_0 \in \{0, \dots, n_\sharp -1\},  &
&k \in \{-k_0, \dots, k_0\}.
\end{align}
\end{comment}
For each $\sharp = \rL, \rR$ and each $z \in \T,$ let $\hat{A}(\sharp, z) = (\hat{A}_{ij}(\sharp, z))_{ij}$ be the square matrix of dimension $n \times n_\sharp $ defined by the following block-matrix representation;
\begin{align}
\label{equation1: fourier transform of Asharp}
% ----------------- %
\hat{A}(\sharp , z) &:=
\begin{pmatrix}
\hat{A}_{11}(\sharp , z) & \dots & \hat{A}_{1n}(\sharp , z) \\
\vdots & \ddots & \vdots \\
\hat{A}_{1n}(\sharp , z) & \dots & \hat{A}_{nn}(\sharp , z) 
\end{pmatrix}, \\
% ----------------- %
\label{equation2: fourier transform of Asharp}
\hat{A}_{ij}(\sharp , z) &:=
\sum_{k=-k_0}^{k_0} 
\begin{pmatrix}
a_{ij}^{k}(\sharp , 0) & 0  &\dots & 0 \\
0 & a_{ij}^{k}(\sharp , 1)  & \dots & 0 \\
\vdots &  \vdots & \ddots & \vdots\\
0 &  \dots & \dots & a_{ij}^{k}(\sharp , n_\sharp  - 1) 
\end{pmatrix}
\begin{pmatrix}
0  &   &              &\\
\vdots  &   &\bigidentity  & \\
0  &   &              &\\ 
z  & 0 &\dots         &0  
\end{pmatrix}^k, 
\end{align}
where $\bf{1}$ denotes the identity matrix of dimension $n_\sharp  - 1.$ If $A$ is a strictly local operator of the form \cref{equation: strictly local operator in simple form}, then the following the following assertions hold true:
\begin{enumerate}[(i)]
\item We have that $A$ is Fredholm if and only if $\T \ni z \longmapsto \det \hat{A}(z,\sharp ) \in \C$ is nowhere vanishing on $\T$ for each $\sharp  = \rL, \rR.$ In this case, the Fredholm index of $A$ is given by \cref{equation: index and winding number}.
\item The essential spectrum of $A$ is given by \cref{equation: essential spectrum of A}.
\end{enumerate}
\end{mtheorem}

In general, the Fredholm index and essential spectrum are meaningful only in infinite dimensions. Note, however, that \cref{theorem: topological invariants of strictly local operators with asymptotically periodic parameters} allows us to fully classify these two topological invariants for a strictly local operator in the language of linear algebra. In terms of practical applications, \cref{theorem: topological invariants of strictly local operators with asymptotically periodic parameters} can be applied to the time-evolution operator of a discrete-time quantum walk defined on the integer lattice $\Z,$ provided that it is an operator of the form \cref{equation: strictly local operator in simple form} satisfying the asymptotically periodic assumptions \crefrange{equation: asymptotical left limit}{equation: asymptotical right limit}.

\begin{remark}
\cref{theorem: Tanaka2020} is a special case of \cref{theorem: topological invariants of strictly local operators with asymptotically periodic parameters}. Indeed, with the notation introduced in \cref{theorem: topological invariants of strictly local operators with asymptotically periodic parameters}, if $n_\sharp  = 1$ for each $\sharp = \rL, \rR,$ then \crefrange{equation: asymptotical left limit}{equation: asymptotical right limit} become:
\[
a_{ij}^k(\rL, 0) := \lim_{x \to -\infty} a_{ij}^k(x), \qquad 
a_{ij}^k(\rR, 0) := \lim_{x \to +\infty} a_{ij}^k(x).
\]
In this case, we show that \cref{equation1: fourier transform of Asharp} is given by \cref{equation: Ahat in simple form}. We define the two matrices $A_{k}(\rL), A_{k}(\rR)$ by \cref{equation: two-side limits of Ay} for each $k;$
\[
A_{k}(\sharp) 
:=
\begin{pmatrix}
a_{11}^{k}(\sharp, 0)  & \dots & a_{1n}^{k}(\sharp, 0) \\
\vdots & \ddots & \vdots \\
a_{n1}^{k}(\sharp, 0)  & \dots &  a_{nn}^{k}(\sharp, 0) 
\end{pmatrix}, \qquad  \sharp = \rL,\rR.
\]

We have $\hat{A}_{ij}(\sharp, z) = \sum_{k=-k_0}^{k_0} a_{ij}^{k}(\sharp, 0) z^k,$ and so
\[
\hat{A}(\sharp, z) =
\begin{pmatrix}
\sum_{k=-k_0}^{k_0} a_{11}^{k}(\sharp, 0) z^k & \dots & \sum_{k=-k_0}^{k_0} a_{1n}^{k}(\sharp, 0) z^k \\
\vdots & \ddots & \vdots \\
\sum_{k=-k_0}^{k_0} a_{n1}^{k}(\sharp, 0) z^k & \dots & \sum_{k=-k_0}^{k_0} a_{nn}^{k}(\sharp, 0) z^k 
\end{pmatrix},
\]
which is consistent with \cref{equation: Ahat in simple form}.
\end{remark}

% ------------------------------------------------------------------------------------------------------------ %
%\subsection{Notation and terminology}
\subsection{Preliminaries}
\label{section: preliminaries}

By operators we shall always mean everywhere-defined bounded linear operators between Banach spaces throughout this paper. An operator $A$ on a Hilbert space $\cH$ is said to be \textbi{Fredholm}, if $\ker A, \ker A^*$ are finite-dimensional and if $A$ has a closed range. Given such $A,$ we define the \textbi{Fredholm index} of $A$ by $\ind(A) := \dim \ker A - \dim \ker A^*.$ It is well-known that the Fredholm index is invariant under compact perturbations. That is, given an operator $A$ on $\cH$ and a compact operator $K$ on $\cH,$ we have that $A$ is Fredholm if and only if so is $A + K,$ and in this case $\ind(A) = \ind(A + K).$ The (Fredholm) \textbi{essential spectrum} of an operator $A$ on $\cH$ is defined as the set $\ess(A)$ of all $\lambda \in \C,$ such that $A - \lambda$ fails to be Fredholm. Note that  $\ess(A)$ is also stable under compact perturbations.

The Hilbert space of all square-summable $\C$-valued sequences $\Psi = (\Psi(x))_{x \in \Z}$ is denoted by the shorthand $\ell^2(\Z) := \ell^2(\Z,\C).$ We have a natural orthogonal decomposition $\ell^2(\Z) = \ell^2_\rL(\Z) \oplus \ell^2_\rR(\Z),$ where
\[
\ell^2_{\rL}(\Z) := \{\Psi \in \ell^2(\Z) \mid \Psi(x) = 0 \,\, \forall x \geq 0 \}, \qquad
\ell^2_\rR(\Z) := \{\Psi \in \ell^2(\Z) \mid \Psi(x) = 0 \,\, \forall x < 0\}.
\]
The orthogonal projections of $\ell^2(\Z)$ onto the above subspaces shall be denoted by $P_{\rL}$ and $P_{\rR} = 1 - P_{\rL}$ respectively. For each $\sharp = \rL,\rR,$ the orthogonal projection $P_{\sharp}$ can be written as $P_{\sharp} = \iota_{\sharp} \iota_{\sharp}^*,$ where $\iota_{\sharp} : \ell^2_{\sharp}(\Z) \hookrightarrow \ell^2(\Z)$ is the inclusion mapping. The \textbi{left-shift operator} $L$ on $\ell^2(\Z)$ is defined by
\begin{equation}
\label{definition: bileteral shift}
L \Psi := \Psi(\cdot + 1), \qquad \Psi \in \ell^2(\Z).
\end{equation}

For each $m \in \N$ any operator $X$ on $\ell^2(\Z, \C^m) := \bigoplus_{j=1}^m \ell^2(\Z)$ admits the following unique block-operator matrix representation;
\begin{equation}
\label{equation: standard representation of X}
X =
\begin{pmatrix}
X_{11} & \dots & X_{1m} \\
\vdots & \ddots& \vdots \\
X_{m1} & \dots & X_{mm} \\
\end{pmatrix}_{\bigoplus_{j=1}^m \ell^2(\Z)},
\end{equation}
where each $X_{ij}$ is an operator on $\ell^2(\Z).$ We shall agree to use the shorthand $X = (X_{ij})$ to mean that \cref{equation: standard representation of X} holds true. With this representation of $X$ in mind, for each $\sharp  = \rL, \rR,$ we define the following compression on $\ell^2_{\sharp}(\Z, \C^m) := \bigoplus_{j=1}^m \ell^2_{\sharp}(\Z);$
\begin{equation}
\label{equation: matrix repesentation of contraction}
X_{\sharp} :=
\begin{pmatrix}
\iota_{\sharp}^* X_{11}\iota_{\sharp} & \dots & \iota_{\sharp}^* X_{1m}\iota_{\sharp} \\
\vdots & \ddots& \vdots \\
\iota_{\sharp}^* X_{m1}\iota_{\sharp} & \dots & \iota_{\sharp}^* X_{mm}\iota_{\sharp} \\
\end{pmatrix}_{\bigoplus_{j=1}^m \ell_\sharp^2(\Z)}.
\end{equation}

For each $m \in \N$ the operator $\tau_{m} : \bigoplus_{j=1}^{m}\ell^2(\Z) \to \ell^2(\Z)$ is defined as the \textit{inverse} of the following unitary operator
\begin{equation}
\ell^2(\Z) \ni \psi \longmapsto 
\begin{pmatrix}
\psi(m \boldsymbol{\cdot}) \\
\vdots \\
\psi(m \boldsymbol{\cdot} + m -1)
\end{pmatrix}
\in \bigoplus_{j=1}^{m}\ell^2(\Z).
\end{equation}
In particular, $\tau_1$ is the identity operator on $\ell^2(\Z).$ Similarly, for each $\sharp = \rL, \rR$ and each $m \in \N$ we define the operator $\tau_{\sharp,m} : \bigoplus_{j=1}^{m}\ell_\sharp^2(\Z) \to \ell_\sharp^2(\Z)$ by
\[
\tau_{\sharp,m} := \iota_{\sharp}^* \tau_{m} \left(\bigoplus_{j=1}^{m} \iota_{\sharp} \right).
\]
It is easy to see that $\tau_{\sharp,m}$ is a unitary operator, since its inverse $\tau_{\sharp}^* = \left(\bigoplus_{j=1}^{m} \iota_{\sharp}^* \right) \tau_{m}^* \iota_{\sharp}$ is given explicitly by the following formula;
\begin{equation}
\label{equation: definition of adjoint of tausharp}
\ell_\sharp^2(\Z) \ni \psi \longmapsto 
\begin{pmatrix}
\psi(m \boldsymbol{\cdot} ) \\
\vdots \\
\psi(m \boldsymbol{\cdot} + m -1)
\end{pmatrix}
\in \bigoplus_{j=1}^{m}\ell_\sharp^2(\Z),
\end{equation}
where $\psi(m \boldsymbol{\cdot} ), \dots, \psi(m \boldsymbol{\cdot} + m - 1) \in \ell^2_\sharp(\Z).$

\begin{lemma}
\label{lemma: unitary transform of important operators}
If $a = (a(x))_{x \in \Z}$ is a bounded $\C$-valued sequence, identified with the associated multiplication operator on $\ell^2(\Z),$ then for each $n \in \N$ we have
\begin{align}
\label{equation: unitary transform of multiplication operator}
\tau_{n}^* a \tau_{n} &= 
\begin{pmatrix}
a(n \boldsymbol{\cdot}) & \dots & 0 \\
\vdots & \ddots & \vdots \\
0 & \dots & a(n \boldsymbol{\cdot} + n - 1) \\
\end{pmatrix}, \\
% ---------- %
\label{equation: unitary transform of shift operator}
\tau_{n}^* L \tau_{n} &= 
\begin{pmatrix}
0  &   &              &\\
\vdots  &   &\bigidentity  & \\
0  &   &              &\\ 
L  & 0 &\dots         &0  
\end{pmatrix},
\end{align}
where $\mathbf{1}$ is the identity operator on $\bigoplus_{j=1}^{n-1}\ell^2(\Z).$ 
\end{lemma}

%What about considering $\tau_{n}^* aL \tau_{n}$ instead (I'm not sure if this idea actually works)\dangernote?
\begin{proof}
For each $\psi \in \ell^2(\Z)$ we have
\begin{align*}
\tau_n^* a \psi 
&= 
\begin{pmatrix}
a(n \boldsymbol{\cdot})  \psi(n \boldsymbol{\cdot}) \\
\vdots \\
a(n \boldsymbol{\cdot} + n-1) \psi(n \boldsymbol{\cdot} + n-1)
\end{pmatrix}
=
\begin{pmatrix}
a(n \boldsymbol{\cdot}) & \dots & 0 \\
\vdots & \ddots & \vdots \\
0 & \dots & a(n \boldsymbol{\cdot} + n - 1) \\
\end{pmatrix}
\tau_n^* \psi, \\
% ------------ %
\tau_n^* L \psi
&= 
\begin{pmatrix}
\psi(n \boldsymbol{\cdot} + 1) \\
\vdots \\
\psi(n \boldsymbol{\cdot} + n)
\end{pmatrix}
=
\begin{pmatrix}
0 & 1 & 0 &\dots & 0 \\
0 & 0 & 1 &\dots & 0 \\
\vdots & \vdots & \vdots & \ddots& \vdots\\
0 & 0 & 0 &\dots & 1 \\
L & 0 & 0 &\dots & 0 
\end{pmatrix}
\tau_n^* \psi.
\end{align*}
The claim follows.
\end{proof}

\begin{corollary}
\label{corollary: unitary transform of projection}
For each $\sharp = \rL, \rR$ and each $n \in \N,$ we have $\tau_{n}^* P_\sharp \tau_{n} = \bigoplus_{j=1}^{n}P_\sharp.$
\end{corollary}
\begin{proof}
For each $\sharp = \rL, \rR,$ we can identify $P_\sharp$ with the multiplication operator $\delta_\sharp.$ It follows from \cref{equation: unitary transform of multiplication operator} that
\[
\tau_{n}^* \delta_\sharp  \tau_{n} = 
\begin{pmatrix}
\delta_\sharp(n \boldsymbol{\cdot}) & \dots & 0 \\
\vdots & \ddots & \vdots \\
0 & \dots & \delta_\sharp(n \boldsymbol{\cdot} + n - 1) \\
\end{pmatrix},
\]
where $\delta_\sharp(n \boldsymbol{\cdot}) = \dots  = \delta_\sharp(n \boldsymbol{\cdot} + n - 1) = \delta_\sharp.$ Therefore, $\tau_{n}^* P_\sharp \tau_{n} = \bigoplus_{j=1}^{n}P_\sharp.$
\end{proof}

In fact, the special case of \cref{theorem: topological invariants of strictly local operators with asymptotically periodic parameters} where $n_\rL = n_\rR$ can be easily proved by making use of \cref{lemma: unitary transform of important operators}. As for the general case $n_\rL \neq n_\rR,$ we require the following non-trivial fact;

\begin{lemma}
\label{lemma: interchange property}
For each $\sharp = \rL, \rR$ and each $m \in \N,$ we have
\begin{equation}
\label{equation: interchange property}
\left(\bigoplus_{j=1}^n \tau_{\sharp,m}^* \right)  A_\sharp \left(\bigoplus_{j=1}^n \tau_{\sharp,m}\right)
= 
\left(\left(\bigoplus_{j=1}^n \tau_{m}^* \right)  A \left(\bigoplus_{j=1}^n \tau_{m}\right) \right)_\sharp.
\end{equation}
More explicitly, the $m \times n$-dimensional strictly local operator $\left(\bigoplus_{j=1}^n \tau_{m}^* \right)  A \left(\bigoplus_{j=1}^n \tau_{m}\right)$ coincides with the block-operator matrix $B(m)$ defined by the following formulas:
\begin{align}
% ----------------- %
B(m) &:=
\begin{psmallmatrix}
B_{11}(m) & \sdots & B_{1n}(m) \\
\svdots & \sddots & \svdots \\
B_{n1}(m) & \sdots & B_{nn}(m) 
\end{psmallmatrix}, \\
% ----------------- %
B_{ij}(m) &:= 
\sum_{k = -k_0}^{k_0} 
\begin{psmallmatrix}
a_{ij}^{k}(m \bigdot) & 0  &\dots & 0 \\
0 & a_{ij}^{k}(m \bigdot + 1)  & \dots & 0 \\
\svdots &  \svdots & \sddots & \svdots\\
0 &  \dots & \dots & a_{ij}^{k}(m \bigdot + m - 1) 
\end{psmallmatrix}
\begin{psmallmatrix}
0  &   &              &\\
\svdots  &   &\bigidentity  & \\
0  &   &              &\\ 
L  & 0 &\sdots         &0  
\end{psmallmatrix}^k, 
\end{align}
where $\bf{1}$ denotes the identity operator of dimension $m - 1.$

\end{lemma} 
The formula \cref{equation: interchange property} shows that $\sharp$-compression and $\tau_m$-unitary transforms can be interchanged. 
\begin{proof}
Note that $A$ can be expresses as a block-operator matrix form $A = (A_{ij})$ according to \cref{equation: matrix repesentation of contraction}, where
\[
A_{ij} := \sum_{k=-k_0}^{k_0} a_{ij}^k L^k.
\]
Note that the left-hand side of \cref{equation: interchange property} becomes;
%We consider the following unitary transform on $\bigoplus_{j=1}^n \ell_\sharp^2(\Z, \C^m);$
\[
\left(\bigoplus_{j=1}^n \tau_{\sharp,m}^* \right)  A_\sharp \left(\bigoplus_{j=1}^n \tau_{\sharp,m}\right)
= 
\begin{psmallmatrix}
\tau_{\sharp,m}^*(\iota_{\sharp}^* A_{11}\iota_{\sharp})\tau_{\sharp,m} & \sdots & \tau_{\sharp,m}^*(\iota_{\sharp} A_{1n}\iota_{\sharp})\tau_{\sharp,m} \\
\svdots & \sddots& \svdots \\
\tau_{\sharp,m}^*(\iota_{\sharp}^* A_{n1}\iota_{\sharp})\tau_{\sharp,m} & \sdots & \tau_{\sharp,m}^*(\iota_{\sharp}^* A_{nn}\iota_{\sharp})\tau_{\sharp,m} \\
\end{psmallmatrix}_{\bigoplus_{j=1}^n \ell_\sharp^2(\Z, \C^m)}.
\]
Note that for each $i,j$ we obtain
\begin{align*}
\tau_{\sharp,m}^*(\iota_{\sharp}^* A_{ij}\iota_{\sharp})\tau_{\sharp,m}  
&= \left(\bigoplus_{j=1}^{m} \iota_{\sharp}^* \right) \tau_{m}^* \iota_{\sharp}(\iota_{\sharp}^* A_{ij}\iota_{\sharp}) \iota_{\sharp}^* \tau_{m} \left(\bigoplus_{j=1}^{m} \iota_{\sharp} \right) \\
&= \left(\bigoplus_{j=1}^{m} \iota_{\sharp}^* \right) \tau_{m}^* P_{\sharp} A_{ij}P_{\sharp} \tau_{m} \left(\bigoplus_{j=1}^{m} \iota_{\sharp} \right) \\
&= \left(\bigoplus_{j=1}^{m} \iota_{\sharp}^* \right) \left(\bigoplus_{j=1}^{m}P_\sharp \right)  \tau_{m}^*A_{ij}\tau_{m} \left(\bigoplus_{j=1}^{m}P_\sharp\right)  \left(\bigoplus_{j=1}^{m} \iota_{\sharp} \right) \\
&= \left(\bigoplus_{j=1}^{m} \iota_{\sharp}^* \right)  \tau_{m}^*A_{ij}\tau_{m}   \left(\bigoplus_{j=1}^{m} \iota_{\sharp} \right),
\end{align*}
where the second last equality follows from  \cref{corollary: unitary transform of projection} and the last equality follows from $\iota_{\sharp}^*\iota_{\sharp} = 1.$ Therefore, \cref{equation: interchange property} holds true. Now,
\[
\tau_{m}^*A_{ij}\tau_{m}
= \sum_{k=-k_0}^{k_0} \tau_{m}^* a_{ij}^k L^k \tau_{m}
= \sum_{k=-k_0}^{k_0} \tau_{m}^* a_{ij}^k \tau_{m} (\tau_{m}^* L \tau_{m})^k.
\]
The claim follows from \cref{lemma: unitary transform of important operators}.
\end{proof}

% ------------------------------------------------------------------------------------------------------------ %
\subsection{Proof of the main theorem}

\begin{lemma}
\label{lemma: transforming to Asharp}
For each $\sharp = \rL, \rR$ let
\begin{align*}
% ----------------- %
{A}(\sharp) &:=
\begin{pmatrix}
{A}_{11}(\sharp) & \dots & {A}_{1n}(\sharp) \\
\vdots & \ddots & \vdots \\
{A}_{1n}(\sharp) & \dots & {A}_{nn}(\sharp) 
\end{pmatrix}, \\
% ----------------- %
{A}_{ij}(\sharp) &:=
\sum_{k=-k_0}^{k_0} 
\begin{pmatrix}
a_{ij}^{k}(\sharp , 0) & 0  &\dots & 0 \\
0 & a_{ij}^{k}(\sharp , 1)  & \dots & 0 \\
\vdots &  \vdots & \ddots & \vdots\\
0 &  \dots & \dots & a_{ij}^{k}(\sharp , n_\sharp  - 1) 
\end{pmatrix}
\begin{pmatrix}
0  &   &              &\\
\vdots  &   &\bigidentity  & \\
0  &   &              &\\ 
L  & 0 &\dots         &0  
\end{pmatrix}^k, 
\end{align*}
where $\bf{1}$ denotes the identity operator of dimension $n_\sharp - 1.$ Then the following operators are compact;
\[
\left(\bigoplus_{j=1}^n \tau_{\sharp,n_\sharp }^* \right)  A_\sharp \left(\bigoplus_{j=1}^n \tau_{\sharp,n_\sharp }\right) - A(\sharp)_\sharp, \qquad \sharp = \rL, \rR.
\]
\end{lemma}
Note that $\hat{A}(\sharp, z)$ defined by \cref{equation2: fourier transform of Asharp} is nothing but the Fourier transform of $A(\sharp).$
\begin{proof}
It follows from \cref{lemma: interchange property} that
\[
\left(\bigoplus_{j=1}^n \tau_{\sharp,n_\sharp}^* \right)  A_\sharp \left(\bigoplus_{j=1}^n \tau_{\sharp,n_\sharp}\right) - A(\sharp)_\sharp
= 
\left(B(n_\sharp) - A(\sharp) \right)_\sharp,
\]
where $B(n_\sharp) - A(\sharp) = (B_{ij}(n_\sharp) - A_{ij}(\sharp))_{i,j}.$ Since $\bigoplus_{j=1}^{n_\sharp} \iota_{\sharp}^* = \bigoplus_{j=1}^{n_\sharp} \iota_{\sharp}^* \bigoplus_{j=1}^{n_\sharp}P_\sharp,$ it remains to show that $C_{ij}(\sharp) := \left(\bigoplus_{j'=1}^{n_\sharp}P_\sharp\right) \left( B_{ij}(n_\sharp) - A_{ij}(\sharp)\right)$ is compact. We obtain
\begin{align*}
C_{ij}(\sharp)
&= 
\sum^{k_0}_{k = -k_0} 
\begin{psmallmatrix}
\delta_{\sharp} (a_{ij}^k(\sharp, 0) - a_{ij}^k(n_\sharp \boldsymbol{\cdot})) & \dots & 0 \\
\svdots & \sddots & \svdots \\
0 & \dots & \delta_{\sharp}(a_{ij}^k(n_\sharp \boldsymbol{\cdot} + n_\sharp - 1)- a_{ij}^k(\sharp,  n_\sharp - 1)) \\
\end{psmallmatrix}
\begin{psmallmatrix}
0  &   &              &\\
\svdots  &   &\bigidentity  & \\
0  &   &              &\\ 
L  & 0 &\sdots         &0  
\end{psmallmatrix}^k,
\end{align*}
where the last equality follows from \cref{lemma: unitary transform of important operators}. Note that each $\delta_\sharp a_{ij}^k(n_\sharp \boldsymbol{\cdot} + j)$ has the following $2$-sided limits:
\begin{align*} 
\lim_{x \to -\infty} \delta_\sharp(x) a_{ij}^k(n_\sharp x + j) 
&= 
\begin{cases}
a_{ij}^k(-\infty, j), & \sharp = \rL, \\
0, & \sharp = \rR, \\
\end{cases} \\
% ------ %
\lim_{x \to +\infty} \delta_\sharp(x) a_{ij}^k(n_\sharp x + j) 
&= 
\begin{cases}
0, & \sharp = \rL, \\
a_{ij}^k(+\infty, j), & \sharp = \rR.
\end{cases}
\end{align*}
It follows that $\delta_{\sharp}(x) (a_{ij}^k(\sharp, j) - a_{ij}^k(n_\sharp x + j)) \to 0$ as $x \to \pm \infty.$ The claim follows.
\end{proof}

\begin{proof}[Proof of \cref{theorem: topological invariants of strictly local operators with asymptotically periodic parameters}]
Note that $A - A_\rL \oplus A_\rR$ is finite rank by \cite[Corollary 2.2]{Tanaka-2020}. Since the Fredholm index and essential spectrum are invariant under compact perturbations, it suffices to consider $A' := A_\rL \oplus A_\rR$ from here on. It follows from \cref{lemma: transforming to Asharp} that the following difference is compact;
\[
\left(\bigoplus_{j=1}^n \tau_{\sharp,n_\sharp }^* \right)  A_\sharp \left(\bigoplus_{j=1}^n \tau_{\sharp,n_\sharp }\right) - A(\sharp)_\sharp.
\]

(i) Since the Fredholmness is invariant under unitary transforms and compact perturbations, we have that $A'$ is Fredholm if and only if $A(\rL)_\rL, A(\rR)_\rR$ are Fredholm. In this case, 
\[
\ind A' = \ind A(\rL)_\rL + \ind A(\rR)_\rR.
\]
On the other hand, it follows from \cite[Theorem 2.4]{Tanaka-2020}(i) that for each $\sharp = \rL, \rR$ the operator $A(\sharp)_\sharp$ is Fredholm if and only if $\T \ni z \longmapsto \det \hat{A}(\sharp, z) \in \C$ is nowhere vanishing on $\T.$ In this case, the Fredholm index of $A(\sharp)_\sharp$ is given by 
\[
\ind A(\sharp)_\sharp = 
\begin{cases}
\wn \left(\det \hat{A}(+\infty, \cdot) \right), & \sharp = \rR, \\
-\wn \left(\det \hat{A}(-\infty, \cdot) \right), & \sharp = \rL.
\end{cases}
\]
The claim follows.

(ii) Since the essential spectrum is invariant under unitary transforms and compact perturbations, we have 
\[
\ess(A') = \ess(A(\rL)_\rL) \cup \ess(A(\rR)_\rR) = \bigcup_{z \in \T}  \sigma \left(\hat{A}(\rR, z) \right) \cup \bigcup_{z \in \T}  \sigma \left(\hat{A}(\rL, z) \right),
\]
where the last equality follows from \cref{equation: index and winding number}.
\end{proof}

\section{Applications of Theorem A}
\label{section: applications of Theorem A}

The purpose of the current section to generalise the existing index formula \cref{equation2: GW indices for ssqw with periodic parameters}.

% ------------------------------------------------------------------------------------------------------------ %
\subsection{Two main theorems}
\label{section: statement of the main theorems}

We first give a brief description of the existing index theory for chirally symmetric unitary operators here (see, for example, \cite{Matsuzawa-Seki-Tanaka-2021} or \cite[\textsection 3.1]{Tanaka-2020}). Let $\cH$ be a Hilbert space, and let $(\varGamma, U)$ be a pair of a unitary self-adjoint operator $\varGamma : \cH \to \cH$ and a unitary operator $U: \cH \to \cH,$ satisfying the chiral symmetry condition \cref{equation: chiral symmetry}. It can then be shown that the real part $R := (U + U^*)/2$ and imaginary part $Q := (U + U^*)/2$ of $U$ admit the following block-operator matrix representations:
\begin{align}
% ------------------------------ %
\label{equation: representation of R and Q}
R =  
\begin{pmatrix}
R_1 & 0 \\
0 & R_2
\end{pmatrix}_{\ker(\varGamma - 1) \oplus \ker(\varGamma + 1)}, \qquad 
Q = \begin{pmatrix}
0 & Q_2 \\
Q_1 & 0
\end{pmatrix}_{\ker(\varGamma - 1) \oplus \ker(\varGamma + 1)},
\end{align}
where the first equality follows from the commutation relation $[\varGamma, R] := \varGamma R - R \varGamma = 0,$ whereas the second equality follows from the anti-commutation relation $\{\varGamma, Q\} := \varGamma Q + Q \varGamma = 0$ (see \cite[Lemma 2.2]{Suzuki-2019} for details). Since $R, Q$ are self-adjoint, we have $R_j^* = R_j$ for each $j=1,2,$ and $Q_2 = Q_1^*.$ It follows that the unitary operator $U = R + iQ$ admits the following representation; 
\begin{equation}
\label{equation: standard representation of U} 
U = 
\begin{pmatrix}
R_1 & iQ_2 \\
iQ_1 & R_2
\end{pmatrix}_{\ker(\varGamma - 1) \oplus \ker(\varGamma + 1)}.
\end{equation}
With \cref{equation: standard representation of U} in mind, we introduce the following formal indices:
\begin{align}
\label{equation: definition of pm indices}
\ind_\pm(\varGamma, U) &:= \dim \ker (R_1 \mp 1) - \dim \ker (R_2 \mp 1), \\
\ind(\varGamma, U) &:= \dim \ker Q_1 - \dim \ker Q_2.
\end{align}
If $\pm 1 \notin \ess(U),$ then $\ind_\pm(\varGamma, U)$ is a well-defined integer, and \cref{equation: topological protection of bounded states} holds true. We prove the following index formula in this section;

%We shall make use of the double-sign correspondence from here on.

\begin{mtheorem}
\label{theorem: index formula for periodic case}
Let $(\varGamma, U) = (\varGamma_{\textnormal{suz}} ,U_{\textnormal{suz}})$ be defined by \crefrange{equation: suzuki split-step quantum walk}{equation: definition of Suzuki evolution operator}.
Suppose that there exist $n_{-\infty}, n_{+\infty} \in \N$ with the property that the following limits exist for each $\star = -\infty, +\infty;$
\begin{equation}
\label{equation: asymptotically periodic limits}
\zeta(\star, m) := \lim_{x \to \star} \zeta(n_\star \cdot x + m), \qquad \zeta \in \{p, a\}, \, m \in \{0, \dots, n_\star -1\}.
\end{equation}

\begin{enumerate}[(i)]
\item Then $\pm 1 \notin \ess(U)$ if and only if for each $\star = -\infty, +\infty$
\begin{equation}
\label{equation: Fredholm characterisation}
\prod_{m=0}^{n_\star  - 1} (1 + p(\star, m)) (1 \mp a(\star, m)) \neq \prod_{m=0}^{n_\star  - 1}  (1 - p(\star, m))(1 \pm a(\star, m)).
\end{equation}

\item Let us impose the following condition;
\begin{equation}
\label{equation1: necessary condition for the Fredholm characterisation}
%\prod_{m=0}^{n_\star  - 1} (1 + \zeta(\star, m)) + \prod_{m=0}^{n_\star  - 1}  (1 - \zeta(\star, m)) > 0, \qquad \zeta = p, a, \qquad \star = -\infty, +\infty.
\prod_{m=0}^{n_\star  - 1} (1 + p(\star, m)) (1 \mp a(\star, m)) + \prod_{m=0}^{n_\star  - 1}  (1 - p(\star, m))(1 \pm a(\star, m)) > 0.
\end{equation}
For each $\star = -\infty, +\infty$ let $p(\star), a(\star) \in [-1,1]$ be uniquely defined through the following formula:
\begin{equation}
\label{equation: definition of p infty}
\frac{\prod_{m=0}^{n_\star  - 1} \left(1 + \zeta(\star, m)\right)}{\prod_{m=0}^{n_\star  - 1}  \left(1 - \zeta(\star, m)\right)}
= \left(\frac{1 + \zeta(\star)}{1 - \zeta(\star)} \right)^{n_\star }, \qquad \zeta = p, a, \qquad \star = -\infty, +\infty.
\end{equation}
Then $\pm 1 \notin \ess(U)$ if and only if $p(\pm \infty) \neq \pm a(\pm \infty).$ Moreover, in this case, we have the following formula;
\begin{equation}
\label{equation: GW indices for ssqw with periodic parameters}
\ind_\pm(\varGamma,U) = 
\frac{\sgn(p(+\infty) \mp a(+\infty)) - \sgn(p(-\infty) \mp a(-\infty))}{2} \in \{-1,0,1\}.
\end{equation}
\end{enumerate}
\end{mtheorem}

We shall make use of the following arithmetic convention for each $r \in (0,\infty];$
\begin{align*}
&r + \infty  = \infty + r = \infty, &
&r \cdot \infty  = \infty \cdot r = \infty, &
&0^{-1} = \infty, &
&\infty^{-1} = 0,
\end{align*}
where $0 \cdot \infty, \infty \cdot 0$ are left undefined throughout this paper. With this convention in mind, we have the homeomorphism $[0,\infty] \ni s \longmapsto s^{-1} \in [0,\infty],$ where the extended half-line $[0,\infty]$ is viewed as a metric space in the obvious way.

It follows from \cref{theorem: index formula for periodic case}(i) that the assumption \cref{equation1: necessary condition for the Fredholm characterisation} is a necessary condition for $\pm 1 \notin \ess(U).$ Note that the assumption \cref{equation1: necessary condition for the Fredholm characterisation} ensures that the left-hand side of \cref{equation: definition of p infty} is a well-defined number in $[0,\infty],$ since the problematic case $0/0$ never occurs. Note also that $\zeta(\star)$ can be indeed uniquely defined through \cref{equation: definition of p infty}, since we have another homeomorphism $\Lambda : [-1,1] \to [0, \infty]$ defined by
\[
\Lambda(s) := \frac{1 + s}{1 - s}, \qquad s \in [-1,1].
\]
%where $\Lambda(1) = 2/0 = \infty.$ 
The function $s \longmapsto \Lambda(s)$ increases from $\Lambda(-1) = 0$ to $\Lambda(+1) = \infty$ as in the following figure;
\begin{figure}[H]
\centering
\label{graph: graph of g}
\begin{tikzpicture}[scale=0.7]
\begin{axis}[axis lines=center,width = 0.9\textwidth,height = 0.5\textwidth,xlabel=$s$, xlabel style={anchor = west}, ylabel=$t$, ylabel style={anchor = south}, xtick= {-1, -1/2, 0, 1/2, 1}, extra x ticks={0}, xticklabel style={anchor = north}, xmin= -1.2, xmax=1.2, ymin= 0, ymax=4.5, ytick= {0,1}, yticklabel style = {anchor = north west}, ytick style={draw=none}]
\addplot[samples=200, domain=-0.999:0.999]{sqrt((1 + x)/(1 - x))};
\addlegendentry{$t = \Lambda(s)$}
%\addplot[dashed,samples=200,domain=-0.9:0.9]{sqrt((1 - x)/(1 + x))};
%\addlegendentry{$y = f(-\kappa)$}
\end{axis}
\end{tikzpicture}
\caption{This figure represents the graph of $t = \Lambda(s).$}
\label{figure: graph of Lambda}
\end{figure}
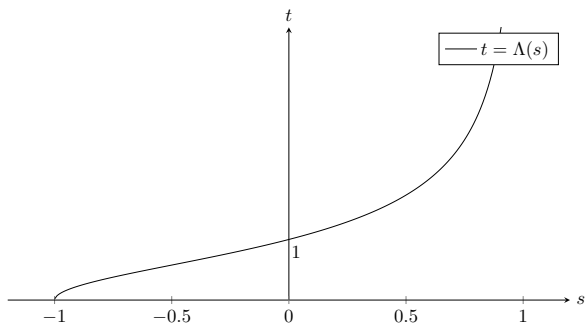
For each $s,s' \in [-1,1],$ we have $\Lambda(-s) = \Lambda(s)^{-1}.$ Furthermore, if $ss' \neq -1,$ then the product $\Lambda(s) \Lambda(s')$ is a well-defined extended non-negative real number, and the following two assertions hold true:
\begin{align}
\label{equation2: properties of Lambda}
&\Lambda(s) \Lambda(s') = \Lambda\left( \frac{s + s'}{1 + ss'}\right), \\
\label{equation3: properties of Lambda}
%&\Lambda(s) \Lambda(s') \lesseqgtr 1 \mbox{ if and only if } s + s' \lesseqgtr 0, \\
&\Lambda(s) \Lambda(s') \lessgtr 1 \mbox{ if and only if } s + s' \lessgtr 0, 
\end{align}
where $1 + ss' > 0$ in \cref{equation2: properties of Lambda}, and where the notation $\lessgtr$ in \cref{equation3: properties of Lambda} simultaneously denotes the three binary relations $>, =, <.$ 

\begin{comment}
and that the following inequalities hold true:
\begin{align}
\label{equation2: necessary condition for the Fredholm characterisation}
\prod_{y=0}^{n_\star  - 1} (1 - p(\star, y)) + \prod_{y=0}^{n_\star  - 1} (1 + p(\star, y)) > 0, \\
% -------------- %
\label{equation3: necessary condition for the Fredholm characterisation}
\prod_{y=0}^{n_\star  - 1} (1 - a(\star, y)) + \prod_{y=0}^{n_\star  - 1} (1 + a(\star, y)) > 0.
\end{align}
\end{comment}

\begin{mtheorem}
\label{theorem: spectrum formula for periodic case}
Let $(\varGamma, U) = (\varGamma_{\textnormal{suz}} ,U_{\textnormal{suz}})$ be defined by \crefrange{equation: suzuki split-step quantum walk}{equation: definition of Suzuki evolution operator}. Suppose that there exist $n_{-\infty}, n_{+\infty} \in \N$ with the property that the limits \cref{equation: asymptotically periodic limits} exist for each $\star = -\infty, +\infty.$
%\begin{equation}
%\zeta(\star, m) := \lim_{x \to \star} \zeta(n_\star \cdot x + m), \qquad \zeta \in \{p, a\}, \, m \in \{0, \dots, n_\star -1\}.
%\end{equation}
Let $\hat{R}_+(\star , z), \hat{R}_-(\star , z)$ be the $n_\star \times n_\star$ matrices defined by the following formula;
\begin{align}
\label{equation: R hat z}
2\hat{R}_\pm(\star , z) 
&:= 
\begin{cases}
r_{1, \pm}(\star , 0)z^{*} + r_{0, \pm}(\star , 0) + r_{1, \pm}(\star , 0)z, & n_\star = 1, \\[10pt]
% ------------ %
\begin{psmallmatrix}
r_{0, \pm}(\star , 0) & r_{1, \pm}(\star , 0) +  r_{1, \pm}(\star , 1)z^*  \\
r_{1, \pm}(\star , 0) + r_{1, \pm}(\star , 1)z & r_{0, \pm}(\star , 1)
\end{psmallmatrix}, & n_\star = 2, \\[10pt]
% ------------ %
\begin{psmallmatrix}
r_{0, \pm}(\star, 0 ) & r_{1, \pm}(\star, 0) & 0 & \cdots & 0 &  r_{1, \pm}(\star, n_\star - 1) z^*\\
r_{1, \pm}(\star, 0) & r_{0, \pm}(\star, 1) &  r_{1, \pm}(\star, 1) & \cdots & 0 & 0 \\
0 & r_{1, \pm}(\star, 1) & r_{0, \pm}(\star, 2) & \cdots & 0 & 0 \\
\svdots & \svdots & \svdots & \sddots & \svdots & \svdots \\
0 & 0 & 0 & \cdots & r_{0, \pm}(\star, n_\star-2) &  r_{1, \pm}(\star, n_\star - 2) \\
r_{1, \pm}(\star, n_\star - 1) z & 0 & 0 & \cdots & r_{1, \pm}(\star, n_\star - 2) & r_{0, \pm}(\star, n_\star-1) 
\end{psmallmatrix}, & n_\star \geq 3,
\end{cases} \\
% -------------- %
r_{0, \pm}(\star, m)  &:= (p(\star, m) \pm 1) a(\star, m) + (p(\star, m) \mp 1) a(\star, m + 1), \\
% -------------- %
r_{1, \pm}(\star, m) &:= \sqrt{(1 \mp p(\star, m))(1 \pm p(\star, m + 1))(1 - a(\star, m + 1)^2)}, 
% -------------- %
\end{align}
where we let $p(\star, n_\star) := p(\star, 0)$ and $a(\star, n_\star) := a(\star, 0).$ Then
\begin{align}
%\label{equation2: essential spectrum of U}
\ess(U) &= \sigma(-\infty) \cup \sigma(+\infty), \\
\label{equation: definition of sigma star}
\sigma(\star)& := \left\{z \in \T \mid \Re z \in \sigma\left(\hat{R}_-(\star , z)\right) \cup \sigma\left(\hat{R}_+(\star , z)\right) \right\}, \qquad \star = -\infty, +\infty.
\end{align}

\end{mtheorem}

% ------------------------------------------------------------------------------------------------------------ %
\subsection{Proof of \texorpdfstring{\cref{theorem: index formula for periodic case}}{Theorem B}}
\label{section: proof of theorem B}

We shall make use of the following remark in what follows;
\begin{remark}
\label{remark: strictly local operator with three coefficients}
Suppose that $A = \alpha_{-1}L^{-1} + \alpha_0 + \alpha_{1}L$ is a one-dimensional strictly local operator, and that there exist $n_{-\infty}, n_{+\infty} \in \N$ with the property that the following limits exist for each $\star = -\infty, +\infty;$
\begin{align*}
%\label{equation: asymptotically periodic limits}
&\alpha_j(\star, m) := \lim_{x \to \star} \alpha_j(n_\star \cdot x + m), &
&j = -1,0,1, & 
&m \in \{0, \dots, n_\star -1\}.
\end{align*}
We introduce the following matrices according to \crefrange{equation1: fourier transform of Asharp}{equation2: fourier transform of Asharp} for each $\star = -\infty, +\infty;$
\begin{align}
\label{equation1: A hat z}
&\hat{A}(\star , z) 
:=
\sum_{j = -1, 0, 1}
\begin{psmallmatrix}
\alpha_{j}(\star, 0) &  0 &0 &0\\
0 &  \alpha_{j}(\star, 1)&0 &0 \\
0 &  0& \sddots &0 \\
0 &  0& 0 & \alpha_{j}(\star,  n_\star  - 1)
\end{psmallmatrix}
\begin{psmallmatrix}
0  &   &              &\\
\svdots  &   &\bigidentity  & \\
0  &   &              &\\ 
z  & 0 &\dots         &0  
\end{psmallmatrix}^j, &
& z \in \T,
\end{align}
where $\bf{1}$ denotes the identity matrix of dimension $n_\star  - 1.$ Note that each $\hat{A}(\star , z)$ admits the following explicit representation;
\begin{equation}
\label{equation2: A hat z}
\hat{A}(\star , z) = 
\begin{cases}
\alpha_{-1}(\star , 0)z^{*} + \alpha_0(\star , 0) + \alpha_{1}(\star , 0)z, & n_\star = 1, \\[10pt]
% ------------ %
\begin{psmallmatrix}
\alpha_{0}(\star , 0) & \alpha_{-1}(\star , 0)z^* + \alpha_{1}(\star , 0)  \\
\alpha_{-1}(\star , 1) + \alpha_{1}(\star , 1)z & \alpha_{0}(\star , 1)
\end{psmallmatrix}, & n_\star = 2, \\[10pt]
% ------------ %
\begin{psmallmatrix}
\alpha_{0}(\star, 0 ) & \alpha_{1}(\star, 0) & 0 & \cdots & 0 &  \alpha_{-1}(\star, 0) z^*\\
\alpha_{-1}(\star, 1) & \alpha_{0}(\star, 1) & \alpha_{1}(\star, 1) & \cdots & 0 & 0 \\
0 & \alpha_{-1}(\star, 2) & \alpha_{0}(\star, 2) & \cdots & 0 & 0 \\
\svdots & \svdots & \svdots & \sddots & \svdots & \svdots \\
0 & 0 & 0 & \cdots & \alpha_{0}(\star, n_\star-2) & \alpha_{1}(\star, n_\star - 2) \\
\alpha_{1}(\star, n_\star - 1) z & 0 & 0 & \cdots & \alpha_{-1}(\star, n_\star-1) & \alpha_{0}(\star, n_\star-1) \\
\end{psmallmatrix}, & n_\star \geq 3.
\end{cases}
\end{equation}
\end{remark}

\begin{lemma}
\label{lemma: wada transform}
If $(\varGamma, U) = (\varGamma_{\textnormal{suz}} ,U_{\textnormal{suz}})$ is defined by \crefrange{equation: suzuki split-step quantum walk}{equation: definition of Suzuki evolution operator}, then there exist two unitary operators $\epsilon, \eta$ on $\ell^2(\Z, \C^2),$ such that the following four decompositions hold true:
\begin{align}
\label{equation1: wada decomposition}
\epsilon^*
\varGamma 
\epsilon
&=
\begin{pmatrix}
1 & 0 \\
0 & -1
\end{pmatrix}, 
&
\epsilon^* \Usuz \epsilon 
&= 
\begin{pmatrix}
1 & 0 \\
0 & -1
\end{pmatrix} (\eta^* \epsilon)^* 
\begin{pmatrix}
1 & 0 \\
0 & -1
\end{pmatrix} (\eta^* \epsilon), \\
% ------------ %
\label{equation2: wada decomposition}
\eta^*
\varGamma'
\eta
&=
\begin{pmatrix}
1 & 0 \\
0 & -1
\end{pmatrix}, 
&
\eta^* \Usuz \eta &= 
(\eta^* \epsilon)
\begin{pmatrix}
1 & 0 \\
0 & -1
\end{pmatrix}
(\eta^* \epsilon)^* 
\begin{pmatrix}
1 & 0 \\
0 & -1
\end{pmatrix}.
\end{align}
More explicitly, if we let $\zeta_\pm := \sqrt{1 \pm \zeta}$ for each $\zeta = p, a,$ then the unitary operators $\epsilon, \eta$ are given respectively by
\begin{equation}
\label{equation: definition of epsilon and eta}
\epsilon :=
\frac{1}{\sqrt{2}}
\begin{pmatrix}
1 & 0 \\
0  &  L^*
\end{pmatrix}
\begin{pmatrix}
p_+  & -p_- \\
p_-  & p_+
\end{pmatrix}, \qquad
\eta :=
\frac{1}{\sqrt{2}}
\begin{pmatrix}
a_+ & -a_- \\
a_-   & a_+
\end{pmatrix}.
\end{equation}
\end{lemma}

\begin{comment}
\begin{lemma}
\begin{equation}
%\label{equation: wada transform}
\epsilon :=
\begin{pmatrix}
1 & 0 \\
0  &  L^*
\end{pmatrix}
\begin{pmatrix}
p_+  & -p_- \\
p_-  & p_+
\end{pmatrix}, \qquad
\eta :=
\begin{pmatrix}
a_+ & -a_- \\
a_-   & a_+
\end{pmatrix}.
\end{equation}
Then the unitary operators $\epsilon, \eta$ give the following diagonalisation:
\begin{equation}
\label{equation: wada decomposition}
\epsilon^*
\varGamma 
\epsilon
=
\begin{pmatrix}
1 & 0 \\
0 & -1
\end{pmatrix} =  
\eta^*
\varGamma'
\eta.
\end{equation}
Moreover,
\[
\epsilon^* \Usuz \epsilon = \begin{pmatrix}
1 & 0 \\
0 & -1
\end{pmatrix} F^* \begin{pmatrix}
1 & 0 \\
0 & -1
\end{pmatrix} F, \qquad 
\eta^* \Usuz \eta = 
F
\begin{pmatrix}
1 & 0 \\
0 & -1
\end{pmatrix}
F^* 
\begin{pmatrix}
1 & 0 \\
0 & -1
\end{pmatrix},
\]
where the unitary operator $F$ is given explicitly by
\[
F:=
\begin{pmatrix}
a_+   & a_- \\
-a_-   & a_+
\end{pmatrix}
\begin{pmatrix}
1 & 0 \\
0  &  L^*
\end{pmatrix}
\begin{pmatrix}
p_+  & -p_- \\
p_-  & p_+
\end{pmatrix}.
\]
\end{lemma}
\end{comment}
\begin{proof}
Note first that we have the following unitary diagonalisation for each $\zeta = p, a$ and each $x \in \Z$ (see, for example, \cite[Example 3.1]{Tanaka-2020});
\[
\begin{pmatrix}
\frac{\zeta_+(x)}{\sqrt{2}} & \frac{-\zeta_-(x)}{\sqrt{2}} \\
\frac{\zeta_-(x)}{\sqrt{2}}   & \frac{\zeta_+(x)}{\sqrt{2}}
\end{pmatrix}^*
\begin{pmatrix}
\zeta(x) & \sqrt{1 - \zeta(x)^2} \\
\sqrt{1 - \zeta(x)^2} & -\zeta(x)
\end{pmatrix}
\begin{pmatrix}
\frac{\zeta_+(x)}{\sqrt{2}} & \frac{-\zeta_-(x)}{\sqrt{2}} \\
\frac{\zeta_-(x)}{\sqrt{2}}   & \frac{\zeta_+(x)}{\sqrt{2}}
\end{pmatrix}
=
\begin{pmatrix}
1 & 0 \\
0 & -1
\end{pmatrix}.
\]
This result motivates us to introduce the unitary operators $\epsilon, \eta$ defined by \cref{equation: definition of epsilon and eta}. Indeed, 
\begin{align*}
\epsilon^* \varGamma \epsilon
=
\begin{pmatrix}
\frac{p_+}{\sqrt{2}}  & \frac{-p_-}{\sqrt{2}}  \\
\frac{p_-}{\sqrt{2}}  & \frac{p_+}{\sqrt{2}} 
\end{pmatrix}^*
\begin{pmatrix}
p  & \sqrt{1 - p ^2} \\
\sqrt{1 - p ^2} & -p 
\end{pmatrix}
\begin{pmatrix}
\frac{p_+}{\sqrt{2}}  & \frac{-p_-}{\sqrt{2}}  \\
\frac{p_-}{\sqrt{2}}  & \frac{p_+}{\sqrt{2}} 
\end{pmatrix}
=
\begin{pmatrix}
1 & 0 \\
0 & -1
\end{pmatrix}, \\
% -------- %
\eta^* \varGamma' \eta
=
\begin{pmatrix}
\frac{a_+}{\sqrt{2}}  & \frac{-a_-}{\sqrt{2}}  \\
\frac{a_-}{\sqrt{2}}    & \frac{a_+}{\sqrt{2}} 
\end{pmatrix}^*
\begin{pmatrix}
a  & \sqrt{1 - a^2} \\
\sqrt{1 - a^2} & -a 
\end{pmatrix}
\begin{pmatrix}
\frac{a_+}{\sqrt{2}}  & \frac{-a_-}{\sqrt{2}}  \\
\frac{a_-}{\sqrt{2}}    & \frac{a_+}{\sqrt{2}} 
\end{pmatrix}
=
\begin{pmatrix}
1 & 0 \\
0 & -1
\end{pmatrix}.
\end{align*}
On the other hand,
\begin{align*}
\epsilon^* \Usuz \epsilon 
&= \epsilon^* \varGamma \varGamma' \epsilon
= (\epsilon^* \varGamma \epsilon)(\epsilon^*\eta)(\epsilon^* \varGamma \epsilon)(\eta^* \epsilon)
= 
\begin{pmatrix}
1 & 0 \\
0 & -1
\end{pmatrix} (\eta^* \epsilon)^* 
\begin{pmatrix}
1 & 0 \\
0 & -1
\end{pmatrix} 
(\eta^* \epsilon), \\
% --------- %
\eta^* \Usuz \eta
&= \eta^* \varGamma \varGamma' \eta
= (\eta^* \epsilon) (\eta^* \varGamma' \eta) (\epsilon^* \eta) (\eta^* \varGamma' \eta)
= (\eta^* \epsilon) 
\begin{pmatrix}
1 & 0 \\
0 & -1
\end{pmatrix} (\eta^* \epsilon)^* 
\begin{pmatrix}
1 & 0 \\
0 & -1
\end{pmatrix}.
\end{align*}
\end{proof}

With the notation introduced in \cref{lemma: wada transform}, it is easy to see that the operator $F := \eta^* \epsilon$ is given explicitly by
\begin{equation}
\label{equation: block representation of F}
F = 
\frac{1}{2}
\begin{pmatrix}
p_+a_+  +  a_- L^*p_-   & - p_-a_+   +  a_-  L^*p_+ \\
-p_+ a_-  + a_+  L^* p_-  & p_-a_-  + a_+ L^*p_+
\end{pmatrix} = :
\begin{pmatrix}
F_{1,-} & F_{2,+} \\
F_{1,+} & F_{2,-}
\end{pmatrix}.
\end{equation}
It follows from \cite[Lemma 3.2]{Cedzich-Geib-Werner-Werner-2021} that $U \mp 1$ is Fredholm if and only if $F_{1,\pm}, F_{2,\pm}$ are Fredholm. In this case, we have
\[
\ind_\pm(\varGamma, U) 
= \ind_\pm(\epsilon^* \varGamma \epsilon, \epsilon^*U \epsilon) 
=
\ind F_{1,\pm} = - \ind F_{2,\pm},
\]
where the first equality follows from the unitary invariance of the indices $\ind_\pm,$ and where the last two equalities follow from \cite[Lemma 3.2]{Cedzich-Geib-Werner-Werner-2021}. Therefore, it remains to compute the Fredholm index of the strictly local operators $F_{1,\pm} = \mp p_+ a_{\mp}  + p_-(\cdot - 1)a_{\pm} L^*$ with the aid of the following lemma;

\begin{lemma}
\label{lemma: winding number of cirlce}
Let $\alpha, \beta \in \R,$ and let $f(z) := \alpha + \beta z^*$ for each $z \in \T.$ Then $f$ is nowhere vanishing if and only if $|\alpha| \neq |\beta|.$ In this case, we have
\[
\wn(f) = 
\begin{cases}
-1, & |\alpha| < |\beta|, \\
0, & |\alpha| > |\beta|.
\end{cases}
\]
\end{lemma}
\begin{proof}
On one hand, if $\beta = 0,$ then the constant function $f = \alpha$ is nowhere vanishing with $\wn(f) = 0$ if and only if $|\alpha| \neq 0.$ On the other hand, if $\beta \neq 0,$ then the image of $f$ is the circle centred at $\alpha$ with the non-zero radius $|\beta|.$ Note that the intersection of the circle and the real line $\R$ is the two-point set  $\{\alpha - |\beta|, \alpha + |\beta|\}.$ It follows that $f$ is nowhere vanishing if and only if $|\alpha| \neq |\beta|,$ where the winding number of $f$ is either $-1$ or $0.$ We have $\wn(f) = -1$ if and only if $\alpha - |\beta| < 0 < \alpha + |\beta|$ if and only if $|\alpha| < |\beta|.$ Similarly, $\wn(f) = 0$ if and only if $|\alpha| > |\beta|.$
\end{proof}

\begin{proof}[Proof of \cref{theorem: index formula for periodic case}]
For each $\zeta = p, a,$ and each $\star = \pm \infty,$ let
\[
\zeta_\pm(\star, m) := \sqrt{1 \pm \zeta(\star, m)}, \qquad m \in \{0, \dots, n_\star  - 1\}.
\]
If we let $f_{0, \pm} := \mp p_+ a_{\mp}$ and $f_{-1, \pm} := p_-(\cdot - 1)a_{\pm},$ then $F_{1,\pm} = f_{-1,\pm} L^* + f_{0,\pm} + 0L.$ For each $z \in \T$ we introduce the following matrix according to \cref{equation2: A hat z};
\begin{equation*}
2\hat{F}_{1,\pm}(\star , z)  
:= 
\begin{cases}
f_{0,\pm}(\star , 0) + f_{-1,\pm}(\star , 0)z^{*}, & n_\star = 1, \\[10pt]
% ------------ %
\begin{psmallmatrix}
f_{0,\pm}(\star , 0)   & f_{-1,\pm}(\star , 0)z^*  \\
f_{-1,\pm}(\star , 1)  & f_{0,\pm}(\star , 1)
\end{psmallmatrix}, & n_\star = 2, \\[10pt]
% ------------ %
\begin{psmallmatrix}
f_{0,\pm}(\star, 0 ) & 0 & 0 & \cdots & 0 &  f_{-1,\pm}(\star, 0) z^*\\
f_{-1,\pm}(\star, 1) & f_{0,\pm}(\star, 1) & 0 & \cdots & 0 & 0 \\
0 & f_{-1,\pm}(\star, 2) & f_{0,\pm}(\star, 2) & \cdots & 0 & 0 \\
\svdots & \svdots & \svdots & \sddots & \svdots & \svdots \\
0 & 0 & 0 & \cdots & f_{0,\pm}(\star, n_\star-2) & 0 \\
0 & 0 & 0 & \cdots & f_{-1,\pm}(\star, n_\star-1) & f_{0,\pm}(\star, n_\star-1) \\
\end{psmallmatrix}, & n_\star \geq 3.
\end{cases}
\end{equation*}

(i) Note first that the following equality holds true, if $n_\star = 1,2;$
\begin{equation}
\label{equation: determinant of hatF}
\det (2\hat{F}_{1,\pm}(\star, z))
= \prod_{m=0}^{n_\star  - 1}  f_{0,\pm}(\star, m) 
+ (-1)^{n_\star  + 1}\left( \prod_{m=0}^{n_\star  - 1}  f_{-1,\pm}(\star, m)\right)z^*.
\end{equation}
In fact, the co-factor expansion easily allows us to prove that \cref{equation: determinant of hatF} also holds true for any $n_\star \geq 3,$ since the determinant of any triangular matrix is the product of its diagonal entries. It follows from \cref{equation: determinant of hatF} that
\begin{align*}
\det (2\hat{F}_{1,\pm}(\star, z))
=
\prod_{m=0}^{n_\star  - 1} \mp p_+(\star, m)a_{\mp}(\star, m)  
+ \left( (-1)^{n_\star  + 1} \prod_{m=0}^{n_\star  - 1}  p_-(\star, m)a_{\pm}(\star, m) \right) z^*,
\end{align*}
since $p'_- := p_-(\cdot - 1)$ satisfies
\[
\prod_{m=0}^{n_\star  - 1}  p'_-(\star, m)
= 
p_-(\star, n_\star - 1)  p_-(\star, 0) 
\dots 
p_-(\star, n_\star - 2)  
= \prod_{m=0}^{n_\star  - 1}  p_-(\star, m).
\]
It follows from \cref{lemma: winding number of cirlce} that $z \longmapsto \det \hat{F}_{1,\pm}(\star, z)$ is nowhere vanishing if and only if \cref{equation: Fredholm characterisation} holds true for each $\star = -\infty,+\infty,$ since
\begin{align*}
\prod_{m=0}^{n_\star  - 1} \mp p_+(\star, m)a_{\mp}(\star, m) &= (\mp 1)^{n_\star} \left(\prod_{m=0}^{n_\star  - 1} (1 + p(\star, m))(1 \mp a(\star, m)) \right)^{1/2}, \\ 
(-1)^{n_\star  + 1}\prod_{m=0}^{n_\star  - 1}  p_-(\star, m)a_{\pm}(\star, m) &= (-1)^{n_\star  + 1}  \left(\prod_{m=0}^{n_\star  - 1}  (1 - p(\star, m))(1 \pm a(\star, m)) \right)^{1/2}.
\end{align*}
Moreover, if \cref{equation: Fredholm characterisation} holds true, then $w_\pm(\star) := \wn (\det \hat{F}_{1,\pm}(\star, \cdot))$ is given by
\begin{equation}
\label{equation: baby index formula}
w_\pm(\star) =
\begin{cases}
-1, & \prod\limits_{m=0}^{n_\star  - 1} (1 + p(\star, m)) (1 \mp a(\star, m)) < \prod\limits_{m=0}^{n_\star  - 1}  (1 - p(\star, m))(1 \pm a(\star, m)), \\
0, & \prod\limits_{m=0}^{n_\star  - 1} (1 + p(\star, m)) (1 \mp a(\star, m)) > \prod\limits_{m=0}^{n_\star  - 1}  (1 - p(\star, m))(1 \pm a(\star, m)).
\end{cases}
\end{equation}

(ii) Suppose that \cref{equation1: necessary condition for the Fredholm characterisation} holds true. Note that
\begin{align*}
%\label{equation: definition of xi infty}
\Lambda(-\zeta(\star))^{n_\star }
= \left(\Lambda(\zeta(\star))^{n_\star }\right)^{-1}
=
\left(
\frac{\prod_{m=0}^{n_\star  - 1} (1 + \zeta(\star, m))}{\prod_{m=0}^{n_\star  - 1}  (1 - \zeta(\star, m))}
\right)^{-1}
=
\frac{\prod_{m=0}^{n_\star  - 1} (1 - \zeta(\star, m))}{\prod_{m=0}^{n_\star  - 1}  (1 + \zeta(\star, m))}.
\end{align*}

We consider
\begin{equation}
\label{equation: binary relation}
\prod_{m=0}^{n_\star  - 1} (1 + p(\star, m)) \prod_{m=0}^{n_\star  - 1} (1 \mp a(\star, m))  \lessgtr \prod_{m=0}^{n_\star  - 1}  (1 - p(\star, m)) \prod_{m=0}^{n_\star  - 1}(1 \pm a(\star, m)),
\end{equation}
where the notation $\lessgtr$ simultaneously denotes the three binary relations $<, =, >.$ On one hand, if $\prod_{m=0}^{n_\star  - 1}  (1 - p(\star, m))(1 \pm a(\star, m)) > 0,$ then \cref{equation: binary relation} is equivalent to
\[
\Lambda(p(\star))^{n_\star } \Lambda(\mp a(\star))^{n_\star } \lessgtr 1
\iff 
\Lambda(p(\star)) \Lambda(\mp a(\star)) \lessgtr 1
\iff p(\star) \mp a(\star) \lessgtr 0,
\]
where the first equivalence follows from the fact that $[0,\infty] \ni s \longmapsto s^{n_\star} \in [0,\infty]$ is an increasing function. On the other hand, if $\prod_{m=0}^{n_\star  - 1} (1 + p(\star, m)) (1 \mp a(\star, m)) \neq 0,$ then \cref{equation: binary relation} is equivalent to
\[
1 \lessgtr \Lambda(-p(\star))^{n_\star} \Lambda(\pm a(\star))^{n_\star}
\iff 0 \lessgtr -p(\star) \pm a(\star) \iff 
p(\star) \mp a(\star) \lessgtr 0.
\]
It follows that \cref{equation: binary relation} is equivalent to $p(\star) \mp a(\star) \lessgtr 0.$ It follows from (i) that $\pm 1 \notin \ess(U)$ if and only if $p(\star) \mp a(\star) \neq 0.$ In this case, \cref{equation: baby index formula} becomes 
\begin{align*}
w_\pm(\star) 
&=
\begin{cases}
-1, & p(\star) \mp a(\star) < 0, \\
0, &  p(\star) \mp a(\star) > 0,
\end{cases} \\
&= \frac{\sgn(p(\star) \mp a(\star)) - 1}{2}.
\end{align*}
We get
\[
\ind_\pm(\varGamma,U)
%= \ind F_{1,\pm} 
= w_\pm(+\infty) - w_\pm(-\infty)
= \frac{\sgn(p(+\infty) \mp a(+\infty)) - \sgn(p(-\infty) \mp a(-\infty))}{2}.
\]
The claim follows.
\end{proof}

% ------------------------------------------------------------------------------------------------------------ %
\subsection{Proof of \texorpdfstring{\cref{theorem: spectrum formula for periodic case}}{Theorem C}}

Note that the evolution operator $\Usuz$ of Suzuki's split-step quantum walk satisfying the asymptotically periodic assumption \cref{equation: asymptotically periodic limits} is a $2$-dimensional strictly local operator. In theory, it is possible to compute $\ess(\Usuz)$ by making use of \cref{theorem: topological invariants of strictly local operators with asymptotically periodic parameters}(ii), but we shall end up with spectral analysis of $2 n_\star \times 2 n_\star$ matrices according to \cref{equation1: fourier transform of Asharp}. In order to reduce the complexity of computations, we may only focus on the real part of $\Usuz$ as the following two lemmas suggest:

\begin{lemma}
Let $(\varGamma, U)$ be any chiral pair on an abstract Hilbert space $\cH,$ and let $R$ be the real part of $U.$ Then
\begin{equation}
\label{equation: essential spectrum of U}
\ess(U) = \left\{z \in \T \mid \Re z \in \ess(R)\right\}.
\end{equation}
\end{lemma}
\begin{proof}
%Let $\sigma'$ be the right hand side of \cref{equation: essential spectrum of U}. 
We shall make use of $\ess(R) = \{\Re z \mid z \in \ess(U)\},$ a simple proof of which can be found in \cite[Lemma 3.6]{Tanaka-2020}. It suffices to prove $\left\{z \in \T \mid \Re z \in \ess(R)\right\} \subseteq \ess(U),$ since the reverse inclusion is obvious. If $z \in \T$ satisfies $\Re z \in \ess(R),$ then there exists $z_0 \in \ess(U),$ such that $\Re z = \Re z_0.$ That is, either $z = z_0$ or $z = z_0^*,$ where $z_0^* \in \ess(U)$ by the chiral symmetry condition \cref{equation: chiral symmetry}. We get $z \in \ess(U)$ in either case. The claim follows.
\end{proof}

\begin{lemma}
With the notation introduced in \cref{lemma: wada transform}, let $R, Q$ be the real and imaginary parts of $\Usuz.$ Then the unitary operator $\epsilon$ gives the following decomposition;
\begin{align}
% ------------- %
\label{equation: first wada decomposition}
\epsilon^* R \epsilon
&=
\begin{pmatrix}
R_{\epsilon_1} & 0 \\
0 & R_{\epsilon_2}
\end{pmatrix}, &
\epsilon^* Q \epsilon
&=
\begin{pmatrix}
0 & Q_{\epsilon_0}^* \\ 
Q_{\epsilon_0} & 0
\end{pmatrix},
\end{align}
where the three operators $R_{\epsilon_1}, R_{\epsilon_2}, Q_{\epsilon_0}$ are defined respectively by
%What about the following notation\dangernote?
%\[
%2R_{\epsilon_1} := p_{1,-}  L p_{1, +} \sqrt{1 - p_2^2}  + p_{1,+} \sqrt{1 - p_2^2}  L^* p_{1,-}  + (1 + p_1) p_2 - (1 - p_1) p_2(\cdot + 1), \\
%\]
\begin{align}
% ------------- %
%\label{equation: definition of Qgamma}
%2i Q_{\gamma_0} &:=  a_+ L^* q a_+ - a_- q L  a_- - |b|(p + p(\cdot - 1)), \\
% ------------- %
\label{equation: definition of Repsilon1}
2R_{\epsilon_1} &:= p_-  L p_+ \sqrt{1 - a^2}  + p_+ \sqrt{1 - a^2}  L^* p_-  + (1 + p) a - (1 - p) a(\cdot + 1), \\
% ------------- %
\label{equation: definition of Repsilon2}
2R_{\epsilon_2} &:= p_+  L p_-\sqrt{1 - a^2}  + p_- \sqrt{1 - a^2}  L^*  p_+  - (1 - p)a + (1 + p) a(\cdot + 1), \\
% ------------- %
\label{equation: definition of Qepsilon}
-2i Q_{\epsilon_0} &:= p_+  L p_+\sqrt{1 - a^2} - p_- \sqrt{1 - a^2}  L^* p_-  - \sqrt{1 - p^2}(a + a(\cdot + 1)).
% ------------- %
%2R_{\gamma_1} &:= a_- L q^* a_+ + a_+ q L  a_- + a_+^2p - a_-^2p(\cdot - 1), \\
% ------------- %
%2R_{\gamma_2} &:= a_+ L^* q a_- + a_- q L a_+ - a_-^2p + a_+^2p(\cdot - 1).
% ------------- %
\end{align}
\end{lemma}
In fact, the formula \cref{equation: definition of Qepsilon} can be found \cite[Lemma 3.2]{Tanaka-2020}. It is possible to prove \crefrange{equation: definition of Repsilon1}{equation: definition of Repsilon2} by an analogous argument as \cite[Remark 3.3(ii)]{Tanaka-2020} suggests, but we shall make use of the half-step operator decomposition \cref{equation: block representation of F} in the below proof.
\begin{proof}

Recall that $F := \eta^* \epsilon$ is given by \cref{equation: block representation of F}. It follows from the second equality in \cref{equation1: wada decomposition} that 
\[
\epsilon^* \Usuz \epsilon 
= 
\begin{pmatrix}
F_{1,-}^*F_{1,-} - F_{1,+}^*F_{1,+} & -(F_{2,-}^*F_{1,+} - F_{2,+}^*F_{1,-})^* \\
F_{2,-}^*F_{1,+} - F_{2,+}^*F_{1,-} & F_{2,-}^*F_{2,-} - F_{2,+}^*F_{2,+}
\end{pmatrix}.
\]
Since $\epsilon^* R \epsilon, \epsilon^* Q \epsilon$ are the real and imaginary parts of $\epsilon^* \Usuz \epsilon$ respectively, we get
\begin{align*}
\epsilon^* R \epsilon
&=
\begin{pmatrix}
F_{1,-}^*F_{1,-} - F_{1,+}^*F_{1,+} & 0 \\
0 & F_{2,-}^*F_{2,-} - F_{2,+}^*F_{2,+}
\end{pmatrix}, \\
\epsilon^* Q \epsilon
&=
\begin{pmatrix}
0 & i(F_{2,-}^*F_{1,+} - F_{2,+}^*F_{1,-})^* \\
-i(F_{2,-}^*F_{1,+} - F_{2,+}^*F_{1,-}) & 0
\end{pmatrix},
\end{align*}
where $2F_{1,\pm} = \mp p_+ a_{\mp}  + a_{\pm} L^* p_-$ and $2F_{2,\pm} = \mp p_- a_{\pm}  + a_{\mp} L^*p_+.$ The claim follows from the following direct computations;
\begin{align*}
&4(F_{1,-}^*F_{1,-} - F_{1,+}^*F_{1,+}) = 4R_{\epsilon_1}, \\
&4(F_{2,-}^*F_{2,-} - F_{2,+}^*F_{2,+}) = 4R_{\epsilon_2}, \\
&4(F_{2,-}^*F_{1,+} - F_{2,+}^*F_{1,-}) = 4iQ_{\epsilon_0}. 
\end{align*}
\end{proof}

\begin{remark}
It immediately follows from \crefrange{equation: essential spectrum of U}{equation: first wada decomposition} that
\begin{equation}
\label{equation2: essential spectrum of U}
\ess(U) = \left\{z \in \T \mid \Re z \in \ess(R_{\epsilon_1}) \cup \ess(R_{\epsilon_2}) \right\},
\end{equation}
where each $R_{\epsilon_j}$ is a one-dimensional strictly local operator, unlike the evolution operator $U$ itself. We are now in a position to apply the argument outlined in \cref{remark: strictly local operator with three coefficients} to $R_{\epsilon_j}.$
\end{remark}

\begin{proof}[Proof of \cref{theorem: spectrum formula for periodic case}]
Note that the two operators $R_{+} := R_{\epsilon_1}$ and $R_{-} := R_{\epsilon_2},$ defined respectively by \crefrange{equation: definition of Repsilon1}{equation: definition of Repsilon2}, are operators of the form $2R_\pm = r_{1, \pm}(\cdot - 1)L^{-1} + r_{0, \pm} + r_{1, \pm}L.$ The formula \cref{equation2: A hat z} motivates us to define $\hat{R}_\pm(\star , z)$ by \cref{equation: R hat z}. It follows from \cref{theorem: topological invariants of strictly local operators with asymptotically periodic parameters}(ii) that
\begin{equation}
\label{equation: essential spectrum of Rpm}
\ess\left(R_\pm \right) = \bigcup_{\star = -\infty,+\infty} \left(\bigcup_{z \in \T}  \sigma \left(\hat{R}_\pm(\star, z) \right) \right).
\end{equation}
We get
\begin{align*}
\ess(U) 
= \left\{z \in \T \mid \Re z \in \ess(R_{-}) \cup \ess(R_{+}) \right\} 
= \bigcup_{\star = -\infty,+\infty} \sigma(\star),
\end{align*}
where the first equality follows from \cref{equation2: essential spectrum of U}, and where the last equality follows from \cref{equation: essential spectrum of Rpm}. 
\end{proof}

% ------------------------------------------------------------------------------------------------------------ %
% New Section                                                                                                  %
% ------------------------------------------------------------------------------------------------------------ %
\section{Discussion}
\label{section: discussion}

% ------------------------------------------------------------------------------------------------------------ %
\subsection{The essential spectrum of the one-dimensional split-step quantum walk}

It is shown in \cref{theorem: spectrum formula for periodic case} that $\ess(U) = \sigma(-\infty) \cup \sigma(+\infty),$ where for each $\star = -\infty, +\infty$ the subset $\sigma(\star)$ of $\T$ is defined by \cref{equation: definition of sigma star}. The purpose of this subsection is to give a further classification of  $\sigma(\star)$ by restricting attention to $n_\star = 1$ and $n_\star = 2.$ To do so, we introduce the following notation for simplicity;
\[
q := \sqrt{1 - p^2}, \qquad b := \sqrt{1 - a^2}.
\]
Given a fixed real number $r_0$ and a compact interval $[r_1, r_2]$ in $\R,$ we let
\[
r_0 + [r_1, r_2] := [r_0 + r_1, r_0 + r_2], \qquad 
r_0 - [r_1, r_2] := [r_0 -r_2, r_0 - r_1].
\]

% ------------------------------------------------------------------------------------------------------------ %
\subsubsection{The asymptotically 1-periodic case}

We focus on the case $n_\star = 1$ first. The following proposition can be found in \cite{Tanaka-2020}, but we give an alternative derivation via \cref{theorem: spectrum formula for periodic case}.

\begin{proposition}[{\cite[Theorem B(ii)]{Tanaka-2020}}]
\label{proposition: essential spectrum of ssqw with 1 period}
With the notation introduced in \cref{theorem: spectrum formula for periodic case} in mind, if $n_\star = 1,$ then we have $\sigma(\star) = \sigma\left(\hat{R}_+(\star ,z)\right) = \sigma\left(\hat{R}_-(\star ,z)\right)$ for each $z \in \T.$ More precisely, 
\begin{equation}
\label{equation: sigma star with period 1}
\sigma(\star) =  \left\{z \in \T \mid \Re z \in I(\star)\right\},
\end{equation}
where the closed subinterval $I(\star)$ of $[-1,1]$ is defined by 
\begin{equation}
\label{equation: spectrum of Rhat with period 1}
I(\star) := p(\star, 0)a(\star, 0) + [- q(\star, 0) b(\star, 0),  q(\star, 0) b(\star, 0)].
\end{equation}
Moreover, $\pm 1 \notin \sigma(\star)$ if and only if $p(\star) \neq \pm a(\star).$ 
\end{proposition}
\begin{proof}
It follows from \cref{equation: R hat z} that if $n_\star = 1,$ then
\begin{align*}
2\hat{R}_\pm(\star , e^{it}) 
%&= r_{1, \pm}(\star , 0)e^{-it} + r_{0, \pm}(\star , 0) + r_{1, \pm}(\star , 0)e^{it} \\
&= r_{0, \pm}(\star , 0) + 2r_{1, \pm}(\star , 0)\cos(t),
\end{align*}
where 
\begin{align*}
r_{0, \pm}(\star , 0) &= (p(\star, 0) \pm 1) a(\star, 0) + (p(\star, 0) \mp 1) a(\star, 0) = 2p(\star, 0)a(\star, 0), \\
r_{1, \pm}(\star , 0) &= \sqrt{(1 \mp p(\star, 0))(1 \pm p(\star, 0))(1 - a(\star, 0)^2)}= q(\star, 0) b(\star, 0).
\end{align*}
It follows that $\hat{R}_+(\star ,e^{it}) = \hat{R}_-(\star ,e^{it}) =: \hat{R}_0(\star ,e^{it})$ for each $t \in [0,2\pi],$ and we get \cref{equation: sigma star with period 1}. It is easy to see that $I(\star)  = [p(\star, 0)a(\star, 0)- q(\star, 0) b(\star, 0),  p(\star, 0)a(\star, 0) + q(\star, 0) b(\star, 0)]$ is a subset of $[-1,1];$
\[
|p(\star, 0)a(\star, 0)| + q(\star, 0) b(\star, 0)
\leq 
\frac{p(\star, 0)^2 + a(\star, 0)^2}{2} + \frac{(1 - p(\star, 0)^2) + (1 - a(\star, 0)^2)}{2}
\leq 1.
\]
It remains to show that $\pm 1 \notin \sigma(\star)$ is equivalent to $p(\star) \neq \pm a(\star),$ but we defer the proof until \cref{remark: two-phase assumption}.
\end{proof}

% ------------------------------------------------------------------------------------------------------------ %
\subsubsection{The asymptotically 2-periodic case}

Next, we focus on the case $n_\star = 2.$
\begin{theorem}
\label{theorem: essential spectrum of ssqw with 2 period}
With the notation introduced in \cref{theorem: spectrum formula for periodic case} in mind, if $n_\star = 2,$ then we have $\sigma(\star) = \bigcup_{z \in \T} \sigma\left(\hat{R}_+(\star ,z)\right) = \bigcup_{z \in \T} \sigma\left(\hat{R}_-(\star ,z)\right)$ for each $z \in \T.$ More precisely, 
\begin{equation}
\label{equation: spectrum of Rhat with period 2}
\sigma(\star) =  \left\{z \in \T \mid \Re z \in I_1(\star) \cup I_2(\star)\right\},
\end{equation}
where each closed subinterval $I_j(\star)$ of $[-1,1]$ is defined by 
\begin{align}
\label{equation: definition of Ij}
I_j(\star) &:= d(\star) + (-1)^j\left[\sqrt{d(\star)^2 + d_1(\star)}, \sqrt{d(\star)^2 + d_2(\star)}\right], \\
d(\star) &:= \frac{(p(\star,0) + p(\star,1))(a(\star,0) + a(\star,1))}{4}, \\
% ---------- %
%d_0(\star) &:= (1 + p(\star,0) p(\star,1))(1 + a(\star,0) a(\star,1)) - 2  \\
% ---------- %
d_j(\star) &:= \frac{2 - (1 + p(\star,0) p(\star,1))(1 + a(\star,0) a(\star,1)) +(-1)^{j} \prod_{m=0,1} q(\star,m)b(\star,m)}{2}.
\end{align}
Furthermore, we have the following assertions:
\begin{enumerate}[(i)]
\item The set $I_j(\star)$ given by \cref{equation: definition of Ij} is a well-defined closed interval in the sense that $0 \leq d(\star)^2 + d_1(\star) \leq d(\star)^2 + d_2(\star).$ Moreover, $I_1(\star)$ lies to the left of $I_2(\star).$

\item We have $\pm 1 \notin \sigma(\star)$ if and only if
\begin{equation}
\label{equation2: Fredholm characterisation}
\prod_{m=0,1} (1 + p(\star, m)) (1 \mp a(\star, m)) \neq \prod_{m=0,1}  (1 - p(\star, m))(1 \pm a(\star, m))
\end{equation}

\item If $\prod_{m=0,1} (1 + p(\star, m)) (1 \mp a(\star, m)) + \prod_{m=0,1}  (1 - p(\star, m))(1 \pm a(\star, m)) > 0,$ then we uniquely define $p(\star), a(\star) \in [-1,1]$ through
\begin{equation}
\label{equation2: definition of p infty and a infty}
\frac{\prod_{m=0}^{n_\star  - 1} \left(1 + \zeta(\star, m)\right)}{\prod_{m=0}^{n_\star  - 1}  \left(1 - \zeta(\star, m)\right)}
= \left(\frac{1 + \zeta(\star)}{1 - \zeta(\star)} \right)^{2}, \qquad \zeta = p, a.
\end{equation}
Then $\pm 1 \notin \sigma(\star)$ if and only if $p(\star) \neq \pm a(\star).$ 

\item The sets $I_1(\star), I_2(\star)$ are singleton sets if and only if $\{p(\star, 0),p(\star, 1),a(\star, 0),a(\star, 1)\}$ contains either $-1$ or $+1.$ In this case, each $I_j(\star)$ is given explicitly by
\[
I_j(\star) = 
\left\{d(\star) + (-1)^j\sqrt{d(\star)^2 + \frac{2 - (1 + p(\star,0) p(\star,1))(1 + a(\star,0) a(\star,1))}{2}} \right\}.
\]
\end{enumerate}
\end{theorem}

\begin{comment}
\begin{conjecture}[Gap conjecture]
(iii) $I_1 \cup I_2$ is connected if and only if $p_0 = p_1$ and $a_0 = a_1.$ 
\end{conjecture}
Prove this result\dangernote.
\begin{proof}
We shall make use of \cref{lemma: continuous matrix-valued function}(iii) that $I_1 \cup I_2$ is connected if and only if the following two equalities hold true:
\begin{align*}
&(1 + p_0) a_0 - (1 - p_0) a_1 = (1 + p_1) a_1 - (1 - p_1) a_0, \\
&(1 + p_1)(1 - p_0)(1 - a_1^2) = (1 + p_0)(1 - p_1)(1 - a_0^2).
\end{align*}
On one hand, if $p_0 = p_1$ and $a_0 = a_1,$ then the above equalities hold true. On the other hand, let the above equalities hold true. The above equations are equivalent to:
\begin{align*}
a_0 - a_1 &= \frac{(p_1 - p_0) (a_0 + a_1)}{2}, \\
\Lambda(p_1)\Lambda(p_0) &= \frac{1 - a_0^2}{1 - a_1^2} = \Lambda(a_1)\Lambda(a_0) \frac{(1 - a_0)^2}{(1 + a_0)^2}.
\end{align*}
\end{proof}
\end{comment}

We show first that \cref{proposition: essential spectrum of ssqw with 1 period} is a special case of \cref{theorem: essential spectrum of ssqw with 2 period}.
\begin{remark}
%[$2$-phase assumption]
\label{remark: two-phase assumption}
With the notation introduced in \cref{theorem: spectrum formula for periodic case} in mind, let $n_\star = 2.$ If $p(\star, 0) = p(\star, 1)$ and if $a(\star, 0) = a(\star, 1),$ then
\begin{align*}
d(\star) &= \frac{(p(\star,0) + p(\star,1))(a(\star,0) + a(\star,1))}{4} = p(\star,0)a(\star,0), \\
% ---------- %
d_j(\star) &= \frac{2 - (1 + p(\star,0)^2)(1 + a(\star,0)^2) +(-1)^{j} (1 - p(\star,0)^2) (1 - a(\star,0)^2)}{2}.
\end{align*}
It follows that $d(\star)^2 + d_1(\star) = 0,$ and that $d(\star)^2 + d_2(\star) = (1 - p(\star,0)^2)(1 - a(\star,0)^2).$ 
\begin{align*}
I_1(\star) &= \left[p(\star,0)a(\star,0) - q(\star,0)b(\star,0), p(\star,0)a(\star,0) \right], \\
I_2(\star) &= \left[p(\star,0)a(\star,0), p(\star,0)a(\star,0) + q(\star,0)b(\star,0) \right].
\end{align*}
Therefore, $I_1(\star) \cup I_2(\star) = \left[p(\star,0)a(\star,0) - q(\star,0)b(\star,0), p(\star,0)a(\star,0) + q(\star,0)b(\star,0) \right]$ coincides with $I(\star)$ given by \cref{equation: spectrum of Rhat with period 1}. It follows from \cref{theorem: essential spectrum of ssqw with 2 period}(ii),(iii) that $\pm 1 \notin \sigma(\star)$ if and only if $p(\star) \neq \pm a(\star).$
\end{remark}

We prove \cref{theorem: essential spectrum of ssqw with 2 period} with the aid of the following lemma;
\begin{lemma}
\label{lemma: continuous matrix-valued function}
Given $\alpha_1, \alpha_2 \in \R$ and $\beta_1, \beta_2 \geq 0,$ let us consider the one-parameter family $\{R(z)\}_{z \in \T}$ of $2 \times 2$ Hermitian matrices defined by the following formula;
\begin{equation}
\label{equation: definition of Rt}
R(z) := 
\frac{1}{2}
\begin{pmatrix}
\alpha_1			       & \beta_1 + \beta_2 z^* \\
\beta_1 + \beta_2 z   & \alpha_2
\end{pmatrix}, \qquad z \in \T.
\end{equation}
For each $j=1,2,$ let
\[
I_j := 
\frac{\alpha_1 + \alpha_2}{4} + (-1)^j
\left[ \frac{\sqrt{(\alpha_1 - \alpha_2)^2 + 4(\beta_1 - \beta_2)^2}}{4},  \frac{\sqrt{(\alpha_1 - \alpha_2)^2 + 4(\beta_1 + \beta_2)^2}}{4}\right].
\]
Then the following assertions hold true:
\begin{enumerate}[(i)]
\item We have $\bigcup_{z \in \T} \sigma(R(z)) = I_1 \cup I_2,$ where $I_1$ lies to the left of $I_2.$
\item The set $I_1 \cup I_2$ is connected if and only if $\alpha_1 = \alpha_2$ and $\beta_1 = \beta_2.$ In this case, we have
$
I_1 \cup I_2 = [\alpha_1/2 - \beta_1, \alpha_1/2 + \beta_1].
$
%\item For each $j=1,2,$ the interval $I_j$ is a singleton set if and only if $\beta_1 \beta_2 = 0.$
\end{enumerate}
\end{lemma}
Note that $I_1, I_2$ are well-defined, since $(\beta_1 - \beta_2)^2 \leq (\beta_1 + \beta_2)^2$ follows from $(\beta_1 \pm \beta_2)^2 = \beta_1^2 \pm2 \beta_1 \beta_2 + \beta_2^2.$ 
\begin{proof}
We shall identify the unit-circle $\T$ with $[0,1]$ through $[0,1] \ni t \longmapsto e^{it} \in \T.$

(i) We have
\begin{align*}
2 \cdot \tr R(t) &= \alpha_1 + \alpha_2, \\
4 \cdot \det R(t) &= \alpha_1 \alpha_2 - |\beta_1 + \beta_2 e^{it}|^2 = \alpha_1 \alpha_2 - (\beta_1^2 + \beta_2^2) - 2 \beta_1 \beta_2 \cos t.
\end{align*}
The eigenvalues of $R(t)$ are given by
\begin{align*}
\lambda_j(t)
%&= \frac{2 \tr(R(t)) + (-1)^j \sqrt{(2\tr R(t))^2 - 4( 4\det R(t))}}{4} \\
&= \frac{\alpha_1 + \alpha_2 + (-1)^j \sqrt{(\alpha_1 - \alpha_2)^2 + 4(\beta_1^2 + 2 \beta_1 \beta_2 \cos t + \beta_2^2)}}{4},
\end{align*}
where $j=1,2.$ We get
\[
\bigcup_{t \in [0,2\pi]} \sigma(R(t)) = \bigcup_{t \in [0,2\pi]} \{\lambda_1(t)\} \cup \bigcup_{t \in [0,2\pi]} \{\lambda_2(t)\}. 
\]
The range of $[0,2\pi] \ni t \longmapsto \sqrt{(\alpha_1 - \alpha_2)^2 + 4(\beta_1^2 + 2 \beta_1 \beta_2 \cos t + \beta_2^2)} \in \R$ is
\[
\left[\sqrt{(\alpha_1 - \alpha_2)^2 + 4(\beta_1 - \beta_2)^2}, \sqrt{(\alpha_1 - \alpha_2)^2 + 4(\beta_1 + \beta_2)^2} \right],
\]
where $(\beta_1 - \beta_2)^2 \leq (\beta_1 + \beta_2)^2.$ Therefore, $\bigcup_{t \in [0,2\pi]} \{\lambda_j(t)\} = I_j$ for each $j=1,2.$ Note also that $I_1$ is located to the left of $I_2;$
\[
\Delta(I_1, I_2) := \min I_2 -  \max I_1 = \sqrt{(\alpha_1 - \alpha_2)^2 + 4(\beta_1 - \beta_2)^2} \geq 0.
\]

(ii) The gap $\Delta(I_1, I_2)$ becomes $0$ if and only if $\alpha_1 - \alpha_2 = 0 = \beta_1 - \beta_2.$ In this case, 
\[
I_1 \cup I_2 
%= \left[\frac{\alpha_1 + \alpha_2 - \sqrt{4(\beta_1 + \beta_2)^2}}{4}, \frac{\alpha_1 + \alpha_2 + \sqrt{4(\beta_1 + \beta_2)^2}}{4}\right]
= [\alpha_1/2 -  \beta_1, \alpha_1/2 + \beta_1].
\]

\begin{comment}
(iii) Since
$
(\beta_1^2 + \beta_2^2) \pm 2 \beta_1 \beta_2 = \beta_1^2 \pm 2 \beta_1 \beta_2   + \beta_2^2 = (\beta_1 \pm \beta_2)^2,
$
we have that $(\beta_1 - \beta_2)^2 = (\beta_1 + \beta_2)^2$ if and only if $\beta_1 \beta_2 = 0.$ It follows that for each $j=1,2,$ the interval $I_j$ is a singleton set if and only if $\beta_1 \beta_2 = 0.$
\end{comment}
\end{proof}

\begin{proof}[Proof of \cref{theorem: essential spectrum of ssqw with 2 period}]
Since $n_\star = 2,$ it follows from \cref{equation: R hat z} that
\begin{equation}
\label{equation: special case of R hat z}
\hat{R}_\pm(\star, z) =
\frac{1}{2} 
\begin{pmatrix}
r_{0, \pm}(\star , 0) & r_{1, \pm}(\star , 0) +  r_{1, \pm}(\star , 1)z^*  \\
r_{1, \pm}(\star , 0) + r_{1, \pm}(\star , 1)z & r_{0, \pm}(\star , 1)
\end{pmatrix}.
\end{equation}
Let us first prove that $\bigcup_{z \in \T} \sigma\left(\hat{R}_\pm(\star ,z)\right)$ does not depend on the choice of $\pm$ in order to show \cref{equation: spectrum of Rhat with period 2}. For each $\zeta = p, q,a,b,$ and each $m = 0,1,$ we write $\zeta_m := \zeta(\star,m)$ for simplicity from here on. With this convention in mind, we let
\begin{align*}
r_{0, \pm}(\star , 0) &= (p_0 \pm 1) a_0 + (p_0 \mp 1) a_1 = :\alpha_{1, \pm} , \\
r_{0, \pm}(\star , 1) &= (p_1 \pm 1) a_1 + (p_1 \mp 1) a_0 = :\alpha_{2, \pm}, \\
r_{1, \pm}(\star , 0) &= \sqrt{(1 \mp p_0)(1 \pm p_1)}b_1 =: \beta_{1, \pm}, \\
r_{1, \pm}(\star , 1) &= \sqrt{(1 \pm p_0)(1 \mp p_1)}b_0 =: \beta_{2, \pm},
\end{align*}
where $\alpha_{1, \pm}, \alpha_{2, \pm} \in \R,$ and where $\beta_{1, \pm}, \beta_{2, \pm} \geq 0.$ It follows that \cref{equation: special case of R hat z} is a special case of \cref{equation: definition of Rt}. We shall make use of the following equalities in order to apply \cref{lemma: continuous matrix-valued function}(i) to $\hat{R}_\pm(\star, z):$
\begin{align*}
\alpha_{1, \pm} + \alpha_{2, \pm} &= (p_0 + p_1)(a_0 + a_1), \\
% -------- %
%\alpha_{1, \pm} - \alpha_{2, \pm} &= (p_0 - p_1)(a_0 + a_1) \pm 2(a_0 - a_1), \\
% -------- %
(\alpha_{1, \pm} - \alpha_{2, \pm})^2
&= (p_0 - p_1)^2(a_0 + a_1)^2  + 4(a_0 - a_1)^2 \pm 4(p_0 - p_1)(a_0^2 - a_1^2), \\
% -------- %
(\beta_{1, \pm} +(-1)^j \beta_{2, \pm})^2
& = (1 - p_0p_1)(2 - a_0^2 - a_1^2)  + 2(-1)^jq_0q_1b_0b_1 \mp (p_0 - p_1)(a_0^2 - a_1^2),
% -------- %
%(\beta_{1, \pm} + \beta_{2, \pm})^2 & = (1 - p_0p_1)(2 - a_0^2 - a_1^2)  + 2q_0q_1b_0b_1 \mp (p_0 - p_1)(a_0^2 - a_1^2), 
\end{align*}
where $j=1,2,$ and where we use $\beta_{1, \pm}^2 + \beta_{2, \pm}^2 = (1 - p_0p_1)(2 - a_0^2 - a_1^2) \mp (p_0 - p_1)(a_0^2 - a_1^2)$ in the last equality. It follows that $d'_{j} := (\alpha_{1, \pm} - \alpha_{2, \pm})^2 + 4(\beta_{1, \pm} +(-1)^j \beta_{2, \pm})^2 \geq 0$ does not depend on the choice of $\pm$ for each $j=1,2.$ Moreover,
\begin{align*}
d'_{j} 
&= (p_0 - p_1)^2(a_0 + a_1)^2  + 4(a_0 - a_1)^2 + 4(1 - p_0p_1)(2 - a_0^2 - a_1^2) + 8(-1)^j q_0q_1b_0b_1 \\
&=(p_0 + p_1)^2(a_0 + a_1)^2 +8(2 - (1+ p_0 p_1)(1+ a_0 a_1) + (-1)^{j} q_0q_1b_0b_1)  \\
&= 16(d(\star)^2 + d_{j}(\star)),
\end{align*}
where the second equality follows from $(p_0 - p_1)^2 = (p_0 + p_1)^2 - 4p_0p_1.$ If we define $I_j(\star)$ according to \cref{equation: definition of Ij}, then it follows from \cref{lemma: continuous matrix-valued function}(i) that 
\[
\bigcup_{z \in \T} \sigma\left(\hat{R}_+(\star ,z)\right) = \bigcup_{z \in \T} \sigma\left(\hat{R}_-(\star ,z)\right) = I_1(\star) \cup I_2(\star).
\]
Note that \cref{equation: spectrum of Rhat with period 2} follows from \cref{theorem: spectrum formula for periodic case}.  

(i) It is obvious that $0 \leq d(\star)^2 + d_1(\star) \leq d(\star)^2 + d_2(\star),$ and that $I_1(\star)$ lies to the left of $I_2(\star).$ It remains to show $I_1(\star) \cup I_2(\star) \subseteq [-1,1].$ Let
\[
d_\pm(\star) := d(\star) \pm \sqrt{d(\star)^2 + d_2(\star)},
\]
where $d_-(\star)$ (resp. $d_+(\star)$) is the minimum (resp. maximum) of $I_1(\star) \cup I_2(\star).$ Let the notation $\leqq$ simultaneously denote $\leq$ and $=.$ We are required to prove that the following equivalent conditions hold true with $|d(\star)| \leq 1$ in mind;
\begin{equation}
\label{equation2: essential spectrum of ssqw with 2 period}
\pm d_\pm(\star) \leqq 1 \mbox{ if and only if } 0 \leqq (1 + p_0p_1)(1 + a_0a_1) \mp (p_0 + p_1)(a_0 + a_1) - q_0q_1b_0b_1,
\end{equation}
where $ (1 + p_0p_1)(1 + a_0a_1) \mp (p_0 + p_1)(a_0 + a_1) \geq 0.$ Indeed,
\begin{align*}
(1 + p_0p_1)&(1 + a_0a_1) \mp (p_0 + p_1)(a_0 + a_1) \\
&=
\begin{cases}
0, & p_0p_1 = -1 \mbox{ or }  a_0a_1 = -1, \\
(1 + p_0p_1)(1 + a_0a_1)\left(1 \mp \frac{p_0 + p_1}{1 + p_0p_1}\frac{a_0 + a_1}{1 + a_0a_1}\right), & \mbox{otherwise.}
\end{cases}
\end{align*}
It remains to prove
\begin{equation}
\label{equation1: essential spectrum of ssqw with 2 period}
0 \leqq ((1 + p_0p_1)(1 + a_0a_1) \mp (p_0 + p_1)(a_0 + a_1))^2 - (1 - p_0^2) (1 - p_1^2)(1 - a_0^2)(1 - a_1^2).
\end{equation}
Let
\begin{align*}
s_\pm &:= \left((1 + p_0p_1)(1 + a_0a_1) \mp (p_0 + p_1)(a_0 + a_1) \right)^2, \\
s'_\pm &:=((1 + p_0p_1)(a_0 + a_1) \mp (1 + a_0a_1)(p_0 + p_1))^2.
\end{align*}
We show that the right hand side of \cref{equation1: essential spectrum of ssqw with 2 period} is $s'_\pm \geq 0;$
\begin{align*}
s_\pm - s'_\pm 
&= (1 + p_0p_1)^2((1 + a_0a_1)^2 - (a_0 + a_1)^2) - (p_0 + p_1)^2((1 + a_0a_1)^2 - (a_0 + a_1)^2)  \\
&= ((1 + p_0p_1)^2 - (p_0 + p_1)^2) ((1 + a_0a_1)^2 - (a_0 + a_1)^2)  \\
%&= (1 + 2p_0p_1 + p_0^2p_1^2 - (p_0^2 + 2p_0p_1 + p_1^2)) (1 + 2a_0a_1 + a_0^2a_1^2 - (a_0^2 + 2a_0a_1 + a_1^2))  \\
%&= (1  + p_0^2p_1^2 - p_0^2  - p_1^2) (1  + a_0^2a_1^2 - a_0^2  - a_1^2)  \\
&= (1 - p_0^2) (1 - p_1^2)(1 - a_0^2)(1 - a_1^2).
\end{align*}
It follows that \cref{equation1: essential spectrum of ssqw with 2 period} holds true. We get $-1 \leq d_-(\star) \leq d_+(\star) \leq 1$ by \cref{equation2: essential spectrum of ssqw with 2 period}, and so $I_1 \cup I_2 \subseteq [-1,1].$

(ii) It follows from (i) that
\begin{equation}
\label{equation3: essential spectrum of ssqw with 2 period}
d_\pm(\star) = \pm 1 \mbox{ if and only if } s'_\pm = 0.
\end{equation}
It follows from a direct computation that 
\[
\prod_{m=0,1} (1 + p_m) (1 \mp a_m) - \prod_{m=0,1}  (1 - p_m)(1 \pm a_m)
= \mp2 ((1 + p_0p_1)(a_0 + a_1) \mp (1 + a_0a_1)(p_0 + p_1)).
\]
Therefore, $\pm 1 \notin \sigma(\star)$ if and only if $\prod_{m=0,1} (1 + p_m) (1 \mp a_m) \neq \prod_{m=0,1}  (1 - p_m)(1 \pm a_m).$

(iii) If $\prod_{m=0,1} (1 + p_m) (1 \mp a_m) + \prod_{m=0,1}  (1 - p_m)(1 \pm a_m) > 0,$ then we define $p(\star),a(\star) \in [-1,1]$ through \cref{equation2: definition of p infty and a infty}. Note that \cref{equation2: Fredholm characterisation} is equivalent to $p(\star) \mp a(\star) \neq 0$ as in the proof of \cref{theorem: index formula for periodic case}(ii). The claim follows from (ii).

(iv) Note that $I_j(\star)$ given by \cref{equation: definition of Ij} is a singleton set if and only if $d_1(\star) = d_2(\star)$ if and only if $\prod_{m=0,1} q(\star,m)b(\star,m) = 0.$ The claim follows.
\end{proof}

% ------------------------------------------------------------------------------------------------------------ %
\subsubsection{The general case}

It is desirable to give a complete classification of $\sigma(\star)$ in full generality. The special cases $n_\star = 1, 2$ we have considered in this subsection are intended as motivating examples for this general approach. It is worth noting that the proof of \cref{theorem: essential spectrum of ssqw with 2 period} is already far from obvious. The general case $n_\star \geq 3$ naturally leads to spectral analysis of Hermitian matrices of the form \cref{equation: almost tridiagonal matrix}, but it is not known to the authors whether or not there is a general standard method for this.

% ------------------------------------------------------------------------------------------------------------ %
\subsection{Exponential decay}

Note that \cref{equation: GW indices for ssqw with periodic parameters} can also be written as \cref{equation2: GW indices for ssqw with periodic parameters}. The following result is also one of the main theorems of the present article;

\begin{theorem}
\label{theorem: index formula with exponential decay for periodic case}
Let $(\varGamma, U) = (\varGamma_{\textnormal{suz}} ,U_{\textnormal{suz}})$ be defined by \crefrange{equation: suzuki split-step quantum walk}{equation: definition of Suzuki evolution operator}. Suppose that there exist $n_{-\infty}, n_{+\infty} \in \N$ with the property that the limits of the form \cref{equation: asymptotically periodic limits} exist for each $\star = \pm \infty,$ and that 
\begin{equation}
\label{equation: supremum assumption}
\sup_{x \in \Z} |\zeta(x)| < 1, \qquad \zeta \in \{p, a\}, \, n_0 \in \{0, \dots, n_\star -1\}.
\end{equation}
Let the four numbers $p(\pm \infty), a(\pm \infty) \in (-1,1)$ be uniquely defined through \cref{equation: definition of p infty}, and let $p(\pm \infty) \neq \pm a(\pm \infty).$ Then the following assertions hold true:
\begin{enumerate}[(i)]
\item We have $\dim \ker(U \mp 1) = |\ind_\pm(\varGamma,U)|,$ where $\ind_\pm(\varGamma,U)$ is given by \cref{equation2: GW indices for ssqw with periodic parameters}.

\item If $(-1)^j(p(-\infty) \mp a(-\infty)) < 0 < (-1)^j(p(+\infty) \mp a(+\infty))$ for some $j=1,2,$ then for any non-zero vector $\Psi_\pm \in \ker(U \mp 1)$ there exists a unique non-zero vector $\psi_\pm \in \ker\left(L + \sqrt{ \Lambda((-1)^j p)\Lambda(\mp (-1)^j a)}\right),$ such that
\begin{equation}
\label{equation: eigenstate for the anosotropic case}
\Psi_\pm =
\begin{pmatrix}
\mp (-1)^j \sqrt{\Lambda(\mp (-1)^j a )} \psi_\pm \\
\psi_\pm
\end{pmatrix}.
\end{equation}
Moreover, the eigenstate $\Psi_\pm$ characterised by \cref{equation: eigenstate for the anosotropic case} exhibits exponential decay in the sense that there exist positive constants $c^\downarrow_\pm, c^\uparrow_\pm, \kappa^\downarrow_\pm, \kappa^\uparrow_\pm, x_\pm,$ such that
\begin{equation}
\label{equation: exponential decay in the anisotropic case}
\kappa^\downarrow_\pm e^{- c^\downarrow_\pm |x|} \leq \|\Psi_\pm(x)\|^2 \leq \kappa^\uparrow_\pm e^{- c^\uparrow_\pm |x|}, \qquad |x| \geq x_\pm.
\end{equation}
\end{enumerate}
\end{theorem}

\begin{remark}
\label{remark: supremum assumption}
\leavevmode
\begin{enumerate}[(i)]
\item Note that \cref{equation: supremum assumption} implies
\[
1 > \sup_{x \in \Z} |\zeta(x)| \geq \limsup_{x \to \infty} |\zeta(\pm x)|.
\]
It follows that $|\zeta(\star, n_0)| < 1$ for each $\zeta \in \{p, a\}$ and each $n_0 \in \{0, \dots, n_\star -1\}.$

\item It is shown in the proof of \cref{theorem: index formula with exponential decay for periodic case} below that the four positive constants $c^\downarrow_\pm, c^\uparrow_\pm, \kappa^\downarrow_\pm, \kappa^\uparrow_\pm$ in \cref{equation: exponential decay in the anisotropic case} can be expressed in terms of $p,a$ (see \crefrange{equation1: refined decay rates}{equation2: refined decay rates} for details).
\end{enumerate}
\end{remark}

We introduce the following lemma in order to prove \cref{theorem: index formula with exponential decay for periodic case}.
\begin{lemma}
\label{lemma: geometric mean lemma}
Let $(\alpha(x))_{x = 0}^\infty$ be a sequence of positive numbers, and let us assume that there exists a natural number $n_0 \in \N$ such that the following limits exist in $(0, \infty):$
\begin{equation}
\label{equation: assumption of periodic limits}
\alpha(+\infty, m) := \lim_{x \to \infty} \alpha(n_0 x + m), \qquad m \in \{0, \dots, n_0-1\}.
\end{equation}
Then $\left(\prod_{m=0}^{x-1} \alpha(m) \right)^{1/x} \to \left(\prod_{m=0}^{n_0-1} \alpha(+\infty, m) \right)^{1/n_0}$ as $x \to \infty.$
\end{lemma}
Note that the special case $n_0 = 1$ is nothing but the well-known result that the geometric mean of a convergent positive sequence converges to its limit. We shall make use of this result in what follows.
\begin{proof}
\begin{comment}
It is a well-known theorem (see, for example, Rudin's book\dangernote%\cite[Theorem 3.37]{Book_Rudin_1976}
) that given a sequence $(\beta(x))_{x \in \N}$ of positive numbers, we have
\[
\liminf_{x \to \infty} \frac{\beta(x+1)}{\beta_j(x)} 
\leq 
\liminf_{x \to \infty} \beta(x)^{1/x}
\leq 
\limsup_{x \to \infty} \beta(x)^{1/x}
\leq 
\limsup_{x \to \infty} \frac{\beta(x+1)}{\beta_j(x)} .
\]
If let $\beta(x) = \left(\prod_{m=1}^x \alpha(m) \right)^{1/x}$ for each $x \in \N,$ then
\[
\liminf_{x \to \infty} \frac{\prod_{m=1}^{x+1} \alpha(m)}{\prod_{m=1}^x \alpha(m)}
\leq \liminf_{x \to \infty} \left(\prod_{m=1}^x \alpha(m) \right)^{1/x} \leq \limsup_{x \to \infty} \left(\prod_{m=1}^x \alpha(m) \right)^{1/x} \leq 
\limsup_{x \to \infty} \frac{\prod_{m=1}^{x+1} \alpha(m)}{\prod_{m=1}^x \alpha(m)},
\]
where 
\[
\frac{\prod_{m=1}^{x+1} \alpha(m)}{\prod_{m=1}^x \alpha(m)} = \alpha(x+1) \to \alpha(\infty,1).
\]
The special case $n_0 = 1$ has been verified. 
\end{comment}
Let $m_0 \in \{1, \dots, n_0\}$ be fixed. If let $\beta(x) = \left(\prod_{m=0}^{x-1} \alpha(m) \right)^{1/x}$ for each $x \in \N,$ then 
\begin{align*}
\beta(n_0 x + m_0)
&= \left(\prod_{m=0}^{n_0 x - 1} \alpha(m) \prod_{m=0}^{m_0-1} \alpha(n_0 x + m) \right) \\
&= \left(\prod_{m=0}^{n_0x - 1} \alpha(m) \right) \left(\prod_{m=0}^{m_0 - 1} \alpha(n_0 x + m) \right) \\
&= \prod_{m=0}^{n_0-1} \left(\prod_{x_m=0}^{x-1} \alpha(x_m n_0 + m)  \right) \left(\prod_{m=0}^{m_0 - 1} \alpha(n_0 x + m) \right),
\end{align*}
where $\left(\prod_{m=0}^{m_0 - 1} \alpha(n_0 x + m) \right)_{x \in \N}$ converges to the positive number $\prod_{m=0}^{m_0 - 1} \alpha(+\infty, m).$ Moreover, we get as $x \to \infty$
\[
\log\left(\prod_{m=0}^{m_0 - 1} \alpha(n_0 x + m) \right)^{\frac{1}{n_0 x + m_0}}
= \frac{\log \prod_{m=0}^{m_0 - 1} \alpha(n_0 x + m)}{n_0 x + m_0} \to 0,
\]
where the last step follows from the fact that $\left(\log \prod_{m=0}^{m_0 - 1} \alpha(n_0 x + m) \right)_{x \in \N}$ is a bounded sequence. It follows from the continuity of the exponential function that 
\begin{equation}
\label{equation1: geometric mean lemma}
\left(\prod_{m=0}^{m_0 - 1} \alpha(n_0 x + m) \right)^{\frac{1}{n_0 x + m_0}} \to e^0 = 1.
\end{equation}
On the other hand, it follows that as $x \to \infty$
\begin{equation}
\label{equation2: geometric mean lemma}
\prod_{m=0}^{n_0-1} \prod_{x_m=0}^{x-1} \alpha(x_m n_0 + m)^{\frac{1}{x}}  
\to \prod_{m=0}^{n_0-1} \alpha(+\infty, m)^{\frac{1}{x_0}}.
\end{equation}
It follows from \crefrange{equation1: geometric mean lemma}{equation2: geometric mean lemma} as $x \to \infty$ we have
\begin{align*}
\beta(n_0 x + m_0)^{\frac{1}{n_0 x + m_0}} 
&=  \left(\prod_{m=0}^{n_0-1} \prod_{x_m=0}^{x-1} \alpha(x_m n_0 + m)^{1/x}  \right)^{\frac{x}{n_0 x + m_0}}  \left(\prod_{m=0}^{m_0 - 1} \alpha(n_0 x + m) \right)^{\frac{1}{n_0 x + m_0}} \\
& \to \left(\prod_{m=0}^{n_0-1} \alpha(+\infty, m)\right)^{\frac{1}{x_0}}.
\end{align*}
It follows that any subsequence of $(\beta(x)^{1/x})_{x \in \N}$ also converges to $\left(\prod_{m=0}^{n_0-1} \alpha(+\infty, m)\right)^{\frac{1}{x_0}},$ since the constant $m_0$ was chosen arbitrarily. The claim follows.
\end{proof}

\begin{proof}[Proof of \cref{theorem: index formula with exponential decay for periodic case}]
For each $j=1,2,$ we introduce the following notation according to \cite[Theorem 3.1]{Matsuzawa-Seki-Tanaka-2021};
\begin{align}
\label{equation: definition of delta jpm}
\delta_{j,\pm}(y) &:=  \sqrt{ \Lambda((-1)^j p(y))\Lambda(\mp (-1)^j a(y))},  \qquad y \in \Z, \\
% -------- %
\label{equation: definition of Delta jpm}
\Delta_{j, \pm} &:= \sum_{x =1}^\infty \left(\prod_{y=0}^{x-1} |\delta_{j,\pm}(-y-1)|^{-2}\right) + \sum_{x =1}^\infty \left(\prod_{y=0}^{x-1} |\delta_{j,\pm}(y)|^2\right).
\end{align}

(i) Note first that the two non-negative numbers $\Delta_{1, \pm}$ and $\Delta_{2, \pm}$ cannot be simultaneously finite, since $|\delta_{1,\pm}(y)\delta_{2,\pm}(y)|^2 = 1$ for each $y \in \Z.$ With this result in mind, \cite[Theorem 3.1]{Matsuzawa-Seki-Tanaka-2021}(i) gives the following classification;
\begin{align}
\label{equation: bulk-edge correspondence for ssqw}
|\ind_\pm(\varGamma,U)| &= \dim \ker(U \mp 1), \\
\label{equation: pm index for ssqw}
\ind_\pm(\varGamma,U) &= 
\begin{cases}
+1, & \Delta_{1, \pm} < \infty, \\
-1, & \Delta_{2, \pm} < \infty, \\
0,  & \Delta_{1, \pm} =  \Delta_{2, \pm} = \infty.
\end{cases}
\end{align}
We are required to show that \cref{equation: pm index for ssqw} agrees with \cref{equation2: GW indices for ssqw with periodic parameters} by making use of the root test. Since the function $\Lambda$ is continuous, for each $\zeta = -p,+p,-a,+a$ and each $\star = \pm \infty,$ the following numbers belong to $(0,\infty);$
\[
\Lambda(\zeta(\star, y)) = \lim_{x \to \star} \Lambda(\zeta(n_\star \cdot x  + y)), \qquad y \in \{0, \dots, n_\star - 1\},
\]
where $|\zeta(\star, n_0)| < 1.$ It follows from \cref{lemma: geometric mean lemma} that as $x \to \infty$
\begin{align*}
\lim_{x \to \infty} \left(\prod_{y=0}^{x-1} \Lambda(\zeta(-y-1))\right)^{\frac{1}{x}} &= \left(\prod_{y=0}^{n_{-\infty}-1} \Lambda(\zeta(-\infty, y)) \right)^{\frac{1}{n_{-\infty}}}, \\
\lim_{x \to \infty} \left(\prod_{y=0}^{x-1} \Lambda(\zeta(y))\right)^{\frac{1}{x}} &= \left(\prod_{y=0}^{n_{+\infty}-1} \Lambda(\zeta(\infty, y)) \right)^{\frac{1}{n_{+\infty}}}.
\end{align*}

Since $|\delta_{j,\pm}(y)|^2 = \left(\Lambda(p(y))\Lambda(\mp a(y)) \right)^{(-1)^{j}}$ for each $y \in \Z,$ we have 
\begin{align*}
&
\lim_{x \to \infty} \left(\prod_{y=0}^{x-1} |\delta_{j,\pm}(-y-1)|^{-2}\right)^{1/x} 
%=  \left(\prod_{y=1}^{x_-} \Lambda(p(-\infty, -y)) \right)^{\frac{1}{x_-}} \left(\prod_{y=1}^{x_-} \Lambda(\mp a(-\infty, -y)) \right)^{\frac{1}{x_-}}
= \left( \Lambda(p(-\infty))\Lambda(\mp a(-\infty)) \right)^{(-1)^{j+1}}, \\
% ----------- %
&
\lim_{x \to \infty}  \left(\prod_{y=0}^{x-1} |\delta_{j,\pm}(y)|^{2}\right)^{1/x} 
= \left(\Lambda(p(+\infty))\Lambda(\mp a(+\infty)) \right)^{(-1)^{j}},
\end{align*}
where $\Lambda(p(\star)) \Lambda(\mp a(\star)) \neq 1$ for each $\star = \pm \infty,$ since we assume $p(\star) \mp a(\star) \neq 0.$ That is, the root test is applicable to each of the two infinite series on the right hand side of \cref{equation: definition of Delta jpm}, and we obtain the following equivalence for each $j=1,2;$
\begin{equation}
\label{equation: iff condition for Delta}
\Delta_{j, \pm} < \infty \mbox{ if and only if }
(-1)^j (p(+\infty) \mp a(+\infty)) < 0 < (-1)^j (p(-\infty) \mp a(-\infty)).
\end{equation}
It is now easy to see that \cref{equation: pm index for ssqw} becomes \cref{equation2: GW indices for ssqw with periodic parameters}.

(ii) Let $\Delta_{j, \pm} < \infty$ for some $j=1,2$ throughout. It follows from \cite[Theorem 3.1]{Matsuzawa-Seki-Tanaka-2021}(ii) that we have the following linear isomorphism;
\begin{equation}
\ker(L - \delta_{j,\pm}) \in \psi \longmapsto  
\begin{pmatrix}
\mp (-1)^j \sqrt{\Lambda(\mp (-1)^j a )} \psi \\
\psi
\end{pmatrix} \in \ker(U \mp 1),
\end{equation}
where $\dim \ker(U \mp 1) = 1.$ In other words, for any non-zero vector $\Psi_\pm \in \ker(U \mp 1)$ there exists a unique non-zero vector $\psi_\pm \in \ker\left(L + \sqrt{ \Lambda((-1)^j p)\Lambda(\mp (-1)^j a)}\right),$ such that $\Psi_\pm$ is given explicitly by \cref{equation: eigenstate for the anosotropic case}. Finally, we introduce the following positive constants to show that $\Psi_\pm$ exhibits exponential decay.
\begin{align}
\delta^{\downarrow}_{j, \pm} &:= \min\left\{\left( \Lambda(p(-\infty))\Lambda(\mp a(-\infty)) \right)^{(-1)^{j+1}},\left(\Lambda(p(+\infty))\Lambda(\mp a(+\infty)) \right)^{(-1)^{j}}\right\}, \\
% ---------- %
\delta^{\uparrow}_{j, \pm} &:= \max\left\{\left( \Lambda(p(-\infty))\Lambda(\mp a(-\infty)) \right)^{(-1)^{j+1}},\left(\Lambda(p(+\infty))\Lambda(\mp a(+\infty)) \right)^{(-1)^{j}}\right\}, \\
% ---------- %
\Lambda_{j, \pm}^\downarrow & := \inf_{x \in \Z} \Lambda(\mp (-1)^j a(x)) + 1, \\
% ---------- %
\Lambda_{j, \pm}^\uparrow & := \sup_{x \in \Z} \Lambda(\mp (-1)^j a(x)) + 1.
\end{align}
Note that $0 < \delta^{\downarrow}_{j, \pm} \leq \delta^{\uparrow}_{j, \pm} < 1,$ where the first inequality follows from \cref{remark: supremum assumption}(i), and where the last inequality follows from \cref{equation: iff condition for Delta} with $\Delta_{j, \pm} < \infty.$ Let $\epsilon > 0$ be small enough, so that $0 < \delta^{\downarrow}_{j,\pm} - \epsilon  < \delta^{\uparrow}_{j,\pm} + \epsilon < 1$ holds true. It then follows from \cite[Theorem 3.1]{Matsuzawa-Seki-Tanaka-2021}(iii) that there exists $x_\pm \in \N,$ such that
\begin{equation}
\label{equation2: exponential decay}
\Lambda_{j, \pm}^\downarrow \left(\delta^{\downarrow} _{j, \pm}- \epsilon \right)^{|x|} \leq \frac{\|\Psi(x)\|^2}{|\psi(0)|^2} \leq  \Lambda_{j, \pm}^\uparrow \left(\delta^{\uparrow}_{j, \pm} + \epsilon\right)^{|x|}, \qquad |x| \geq x_\pm.
\end{equation}
We obtain \cref{equation: exponential decay in the anisotropic case}, if we let 
\begin{align}
\label{equation1: refined decay rates}
\kappa^\downarrow_{j, \pm} &:= |\psi(0)|^2\Lambda^\downarrow_{j, \pm}, & 
c^\downarrow_{j, \pm} &:= -\log\left(\delta^{\downarrow}_{j, \pm} - \epsilon\right), \\
% ------ %
\label{equation2: refined decay rates}
\kappa^\uparrow_{j, \pm} &:= |\psi(0)|^2\Lambda^\uparrow_{j, \pm}, &
c^\uparrow_{j, \pm} &:= -\log\left(\delta^{\uparrow}_{j, \pm} + \epsilon\right).
\end{align}

\begin{remark}
The proof of \cref{theorem: index formula with exponential decay for periodic case} above gives yet another derivation of the index formula \cref{equation2: GW indices for ssqw with periodic parameters} via \cref{equation: pm index for ssqw}. This latter derivation relies only on elementary analysis of first-order difference equations inspired by \cite{Fuda-Funakawa-Suzuki-2018}, while the former derivation outlined in \cref{section: proof of theorem B} makes extensive use of Toeplitz operators. Note, however, that despite its simplicity the latter method alone is insufficient to justify where the technical assumption $p(\pm \infty) \neq \pm a(\pm \infty)$ comes from. It is precisely the language of Toeplitz operators that allows us to establish the non-trivial equivalence of this assumption and the essential gap condition $\pm 1 \notin \ess(U)$ (see \cref{theorem: index formula for periodic case}(ii) for details).
%In summary, \cref{theorem: topological invariants of strictly local operators with asymptotically periodic parameters} plays an indispensable role in establishing the bulk-edge correspondence for the split-step quantum walk (include some references to the theorems \dangernote).
\end{remark}

\end{proof}
\begin{comment}
Note first that
\begin{align*}
%\label{equation: bounded assumption}
&\limsup_{x \in \Z} |p(x)| = \max \{p(-\infty, x_{-\infty}), \dots, p(+\infty, x_{+\infty} - 1)\} < 1, \\
&\limsup_{x \in \Z} |a(x)| = \max \{a(-\infty, x_{-\infty}), \dots, a(+\infty, x_{+\infty} - 1)\} < 1.
\end{align*} 
Thus, ???\dangernote holds true.
%\cref{equation: bounded assumption} 
It follows from \cref{lemma: geometric mean lemma} that as $x \to \infty$
\begin{align*}
&\left(\prod_{y=1}^{x} \Lambda(p(-y))\right)^{\frac{1}{x}} \to \left(\prod_{y=1}^{x_-} \Lambda(p(-\infty, -y)) \right)^{\frac{1}{x_-}}, \\
&\left(\prod_{y=1}^{x} \Lambda(\mp a(-y))\right)^{\frac{1}{x}} \to \left(\prod_{y=1}^{x_-} \Lambda(\mp a( -\infty,-y)) \right)^{\frac{1}{x_-}}, \\
&\left(\prod_{y=0}^{x-1} \Lambda(p(y))\right)^{\frac{1}{x}}  \to \left(\prod_{y=0}^{n_{+\infty} - 1} \Lambda(p(+\infty, y)) \right)^{\frac{1}{n_{+\infty}}}, \\
&\left(\prod_{y=0}^{x-1} \Lambda(\mp a(y))\right)^{\frac{1}{x}}  \to \left(\prod_{y=0}^{n_{+\infty} - 1} \Lambda(\mp a(+\infty, y)) \right)^{\frac{1}{n_{+\infty}}}.
\end{align*}
\end{comment}

\appendix 
\renewcommand{\thesection}{\Alph{section}} %Note that theorem environments look strange without this

%\bibliographystyle{alpha}
%\bibliography{thesis-bibliography} 

\begin{thebibliography}{MRLAG08}

\bibitem[AAKV01]{Aharonov-Ambainis-Kempe-Vazirani-2001}
D.~Aharonov, A.~Ambainis, J.~Kempe, and U.~Vazirani.
\newblock Quantum walks on graphs.
\newblock In {\em Proceedings of the thirty-third annual {ACM} symposium on
  {Theory} of computing - {STOC} '01}, pages 50--59, New York, New York, USA,
  2001. ACM Press.

\bibitem[ABJ15]{Asch-Bourget-Joye-2015}
J.~Asch, O.~Bourget, and A.~Joye.
\newblock Spectral stability of unitary network models.
\newblock {\em Rev. Math. Phys.}, 27(07):1530004, 2015.

\bibitem[ABN{\etalchar{+}}01]{Ambainis-Bach-Nayak-Vishwanath-Watrous-2001}
A.~Ambainis, E.~Bach, A.~Nayak, A.~Vishwanath, and J.~Watrous.
\newblock One-dimensional quantum walks.
\newblock In {\em Proceedings of 33rd ACM Symposium of the Theory of
  Computing}, pages 37--49. ACM Press, 2001.

\bibitem[ADZ93]{Aharonov-Davidovich-Zagury-1993}
Y.~Aharonov, L.~Davidovich, and N.~Zagury.
\newblock Quantum random walks.
\newblock {\em Phys. Rev. A}, 48:1687--1690, 1993.

\bibitem[AFST21]{Asahara-Funakawa-Seki-Tanaka-2020}
K.~Asahara, D.~Funakawa, M.~Seki, and Y.~Tanaka.
\newblock An index theorem for one-dimensional gapless non-unitary quantum
  walks.
\newblock {\em Quantum~Inf.~Process.}, 20(9), 2021.

\bibitem[Amb03]{Ambainis-2003}
A.~Ambainis.
\newblock Quantum walks and their algorithmic applications.
\newblock {\em Int.~J.~Quantum~Inf.}, 1(4):507--518, 2003.

\bibitem[AO13]{Asboth-Obuse-2013}
J.~K. Asbóth and H.~Obuse.
\newblock Bulk-boundary correspondence for chiral symmetric quantum walks.
\newblock {\em Phys.~Rev.~B}, 88(12):121406, 2013.

\bibitem[CGG{\etalchar{+}}18]{Cedzich-Geib-Grunbaum-Stahl-Velazquez-Werner-Werner-2018}
C.~Cedzich, T.~Geib, F.~A. Gr{\"{u}}nbaum, C.~Stahl, L.~Vel{\'{a}}zquez, A.~H.
  Werner, and R.~F. Werner.
\newblock The {Topological} {Classification} of {One}-{Dimensional} {Symmetric}
  {Quantum} {Walks}.
\newblock {\em Ann.~Henri~Poincar{\'{e}}}, 19(2):325--383, 2018.

\bibitem[CGML12]{Cantero-Grunbaum-Moral-Velazquez-2012}
M.~J. Cantero, F.~A. Gr{\"{u}}nbaum, L.~Moral, and L.Vel{\'{a}}zquez.
\newblock One-dimensional quantum walks with one defect.
\newblock {\em Rev. Math. Phys.}, 24(02):1250002, 2012.

\bibitem[CGS{\etalchar{+}}16]{Cedzich-Grunbaum-Stahl-Velazquez-Werner-Werner-2016}
C.~Cedzich, F.~A. Gr{\"{u}}nbaum, C.~Stahl, L.~Vel{\'{a}}zquez, A.~H. Werner,
  and R.~F. Werner.
\newblock Bulk-edge correspondence of one-dimensional quantum walks.
\newblock {\em J.~Phys.~A}, 49(21):21LT01, 2016.

\bibitem[CGS{\etalchar{+}}18]{Cedzich-Geib-Stahl-Velazquez-Werner-Werner-2018}
C.~Cedzich, T.~Geib, C.~Stahl, L.~Vel{\'{a}}zquez, A.~H. Werner, and R.~F.
  Werner.
\newblock Complete homotopy invariants for translation invariant symmetric
  quantum walks on a chain.
\newblock {\em Quantum}, 2:95, 2018.

\bibitem[CGWW21]{Cedzich-Geib-Werner-Werner-2021}
C.~Cedzich, T.~Geib, A.~H. Werner, and R.~F. Werner.
\newblock Chiral floquet systems and quantum walks at half-period.
\newblock {\em Ann.~Henri~Poincar{\'{e}}}, 22(2):375–413, Jan 2021.

\bibitem[FFS17]{Fuda-Funakawa-Suzuki-2017}
T.~Fuda, D.~Funakawa, and A.~Suzuki.
\newblock Localization of a multi-dimensional quantum walk with one defect.
\newblock {\em Quantum~Inf.~Process.}, 16(8), 2017.

\bibitem[FFS18]{Fuda-Funakawa-Suzuki-2018}
T.~Fuda, D.~Funakawa, and A.~Suzuki.
\newblock Localization for a one-dimensional split-step quantum walk with bound
  states robust against perturbations.
\newblock {\em J.~Math.~Phys.}, 59(8):082201, 2018.

\bibitem[FFS19]{Fuda-Funakawa-Suzuki-2019}
T.~Fuda, D.~Funakawa, and A.~Suzuki.
\newblock Weak limit theorem for a one-dimensional split-step quantum walk.
\newblock {\em Rev.~Roumaine~Math.~Pures~Appl.}, 64:157--165, 2019.

\bibitem[FMS{\etalchar{+}}20]{Funakawa-Matsuzawa-Sasaki-Suzuki-Teranishi-2020}
D.~Funakawa, Y.~Matsuzawa, I.~Sasaki, A.~Suzuki, and N.~Teranishi.
\newblock Time operators for quantum walks.
\newblock {\em Lett.~Math.~Phys.}, 110(9):2471–2490, 2020.

\bibitem[Gro96]{Grover-1996}
L.~K. Grover.
\newblock A fast quantum mechanical algorithm for database search.
\newblock In {\em Proceedings of the twenty-eighth annual {ACM} symposium on
  {Theory} of {Computing}}, {STOC} '96, pages 212--219, New York, NY, USA,
  1996. Association for Computing Machinery.

\bibitem[Gud88]{Gudder-1988}
S.~Gudder.
\newblock {\em Quantum Probability}.
\newblock Probability and Mathematical Statistics : a series of monographs and
  textbooks. 1988.

\bibitem[IKK04]{Inui-Konishi-Konno-2004}
N.~Inui, Y.~Konishi, and N.~Konno.
\newblock Localization of two-dimensional quantum walks.
\newblock {\em Phys.~Rev.~A}, 69(5):052323, 2004.

\bibitem[KBF{\etalchar{+}}12]{Kitagawa-Broome-Fedrizzi-Rudner-Berg-Kassal-Aspuru-Demler-White-2012}
T.~Kitagawa, M.~A. Broome, A.~Fedrizzi, M.~S. Rudner, E.~Berg, I.~Kassal,
  A.~Aspuru-Guzik, E.~Demler, and A.~G. White.
\newblock Observation of topologically protected bound states in photonic
  quantum walks.
\newblock {\em Nat.~Commun.}, 3(1), 2012.

\bibitem[Kit12]{Kitagawa-2012}
T.~Kitagawa.
\newblock Topological phenomena in quantum walks: {Elementary} introduction to
  the physics of topological phases.
\newblock {\em Quant.~Inf.~Process.}, 11(5):1107--1148, 2012.

\bibitem[Kon02]{Konno-2002}
N.~Konno.
\newblock Quantum random walks in one dimension.
\newblock {\em Quant.~Inf.~Process.}, 1(5):345--354, 2002.

\bibitem[Kon10]{Konno-2010}
N.~Konno.
\newblock Localization of an inhomogeneous discrete-time quantum walk on the
  line.
\newblock {\em Quant.~Inf.~Process.}, 9(3):405--418, 2010.

\bibitem[KRBD10]{Kitagawa-Rudner-Berg-Demler-2010}
T.~Kitagawa, M.~S. Rudner, E.~Berg, and E.~Demler.
\newblock Exploring topological phases with quantum walks.
\newblock {\em Phys.~Rev.~A}, 82:033429, 2010.

\bibitem[Mat20]{Matsuzawa-2020}
Y.~Matsuzawa.
\newblock An index theorem for split-step quantum walks.
\newblock {\em Quantum~Inf.~Process.}, 19(8), 2020.

\bibitem[Mey96]{Meyer-1996}
D.~A. Meyer.
\newblock From quantum cellular automata to quantum lattice gases.
\newblock {\em J.~Statist. Phys.}, 85(5-6):551–574, 1996.

\bibitem[MKKO20]{Mochizuki-Kim-Kawakami-Obuse-2020}
K.~Mochizuki, D.~Kim, N.~Kawakami, and H.~Obuse.
\newblock Bulk-edge correspondence in nonunitary floquet systems with chiral
  symmetry.
\newblock {\em Phys. Rev. A}, 102(6):062202, 2020.

\bibitem[MKO16]{Mochizuki-Kim-Obuse-2016}
K.~Mochizuki, D.~Kim, and H.~Obuse.
\newblock Explicit definition of {PT} symmetry for nonunitary quantum walks
  with gain and loss.
\newblock {\em Phys.~Rev.~A}, 93(6):062116, 2016.

\bibitem[Mor19]{Morioka-2019}
H.~Morioka.
\newblock Generalized eigenfunctions and scattering matrices for
  position-dependent quantum walks.
\newblock {\em Rev. Math. Phys.}, 31(07):1950019, 2019.

\bibitem[MRLAG08]{Mohseni-Rebentrost-Lloyd-Aspuru-Guzik-2008}
M.~Mohseni, P.~Rebentrost, S.~Lloyd, and A.~Aspuru-Guzik.
\newblock Environment-assisted quantum walks in photosynthetic energy transfer.
\newblock {\em J.~Chem.~Phys.}, 129(17):174106, 2008.

\bibitem[MS19]{Maeda-Suzuki-2019}
M.~Maeda and A.~Suzuki.
\newblock Continuous limits of linear and nonlinear quantum walks.
\newblock {\em Rev.~Math.~Phys.}, 32(04):2050008, 2019.

\bibitem[MSS{\etalchar{+}}18a]{Maeda-Sasaki-Segawa-Suzuki-Suzuki-2018b}
A.~Maeda, H.~Sasaki, E.~Segawa, A.~Suzuki, and K.~Suzuki.
\newblock Weak limit theorem for a nonlinear quantum walk.
\newblock {\em Quantum~Inf.~Process.}, 17(9), 2018.

\bibitem[MSS{\etalchar{+}}18b]{Maeda-Sasaki-Segawa-Suzuki-Suzuki-2018a}
M.~Maeda, H.~Sasaki, E.~Segawa, A.~Suzuki, and K.~Suzuki.
\newblock Scattering and inverse scattering for nonlinear quantum walks.
\newblock {\em Discrete Contin. Dyn. Syst.}, 38(7):3687--3703, 2018.

\bibitem[MSS{\etalchar{+}}19]{Maeda-Sasaki-Segawa-Suzuki-Suzuki-2019}
M.~Maeda, H.~Sasaki, E.~Segawa, A.~Suzuki, and K.~Suzuki.
\newblock Dynamics of solitons for nonlinear quantum walks.
\newblock {\em J.~Phys.~Commun.}, 3(7):075002, 2019.

\bibitem[MST21]{Matsuzawa-Seki-Tanaka-2021}
Y.~Matsuzawa, M.~Seki, and Y.~Tanaka.
\newblock The bulk-edge correspondence for the split-step quantum walk on the
  one-dimensional integer lattice.
\newblock {\em arXiv:2105.06147}, 2021.

\bibitem[NOW21]{Narimatsu-Ohno-Wada-2021}
A.~Narimatsu, H.~Ohno, and K.~Wada.
\newblock Unitary equivalence classes of split-step quantum walks.
\newblock {\em Quantum~Inf.~Process.}, 20(11), 2021.

\bibitem[Ohn16]{Ohno-2016}
H.~Ohno.
\newblock Unitary equivalent classes of one-dimensional quantum walks.
\newblock {\em Quantum~Inf.~Process.}, 15(9):3599–3617, 2016.

\bibitem[Ohn17]{Ohno-2017}
H.~Ohno.
\newblock Unitary equivalence classes of one-dimensional quantum walks ii.
\newblock {\em Quantum~Inf.~Process.}, 16(12), 2017.

\bibitem[OK11]{Obuse-Kawakami-2011}
H.~Obuse and N.~Kawakami.
\newblock Topological phases and delocalization of quantum walks in random
  environments.
\newblock {\em Phys.~Rev.~B}, 84(19):195139, 2011.

\bibitem[PLM{\etalchar{+}}10]{Peruzzo-Lobino-Matthews-Matsuda-Politi-Poulios-Zhou-Lahini-Ismail-Worhoff-Bromberg-Silberberg-Thompson-OBrien-2010}
A.~Peruzzo, M.~Lobino, J.~C.~F. Matthews, N.~Matsud, A.~Politi, K.~Poulios,
  X.~Zhou, Y.~Lahini, N.~Ismail, K.~Wörhoff, Y.~Bromberg, Y.~Silberberg, M.~G.
  Thompson, and J.~L. OBrien.
\newblock Quantum {Walks} of {Correlated} {Photons}.
\newblock {\em Science}, 329(5998):1500--1503, 2010.

\bibitem[Por16]{Portugal-2016}
R.~Portugal.
\newblock Staggered quantum walks on graphs.
\newblock {\em Phys. Rev. A}, 93(6):062335, 2016.

\bibitem[RST17]{Richard-Suzuki-Tiedra-2017}
S.~Richard, A.~Suzuki, and R.~{Tiedra~de~Aldecoa}.
\newblock Quantum walks with an anisotropic coin i: spectral theory.
\newblock {\em Lett.~Math.~Phys.}, 108(2):331–357, 2017.

\bibitem[RST18]{Richard-Suzuki-Tiedra-2018}
S.~Richard, A.~Suzuki, and R.~{Tiedra~de~Aldecoa}.
\newblock Quantum walks with an anisotropic coin ii: scattering theory.
\newblock {\em Lett.~Math.~Phys.}, 109(1):61–88, 2018.

\bibitem[Seg11]{Segawa-2011}
E.~Segawa.
\newblock Localization of quantum walks induced by recurrence properties of
  random walks.
\newblock {\em J.~Comput.~Theor.~Nanos.}, 10, 2011.

\bibitem[ST19a]{Sambou-Tiedra-2019}
D.~Sambou and R.~{Tiedra~de~Aldecoa}.
\newblock Quantum time delay for unitary operators: General theory.
\newblock {\em Rev. Math. Phys.}, 31(06):1950018, 2019.

\bibitem[ST19b]{Suzuki-Tanaka-2019}
A.~Suzuki and Y.~Tanaka.
\newblock The witten index for 1d supersymmetric quantum walks with anisotropic
  coins.
\newblock {\em Quantum~Inf.~Process.}, 18(12), 2019.

\bibitem[Suz16]{Suzuki-2016}
A.~Suzuki.
\newblock Asymptotic velocity of a position-dependent quantum walk.
\newblock {\em Quantum~Inf.~Process.}, 15(1):103–119, 2016.

\bibitem[Suz19]{Suzuki-2019}
A.~Suzuki.
\newblock Supersymmetry for chiral symmetric quantum walks.
\newblock {\em Quantum~Inf.~Process.}, 18(12), 2019.

\bibitem[Tan21]{Tanaka-2020}
Y.~Tanaka.
\newblock A constructive approach to topological invariants for one-dimensional
  strictly local operators.
\newblock {\em J. Math. Anal. Appl.}, 500(1):125072, 2021.

\bibitem[{Tie}20]{Tiedra-2020}
R.~{Tiedra~de~Aldecoa}.
\newblock Stationary scattering theory for unitary operators with an
  application to quantum walks.
\newblock {\em J.~Funct.~Anal.}, 279(7):108704, 2020.

\bibitem[Wad19]{Wada-2019}
K.~Wada.
\newblock Absence of wave operators for one-dimensional quantum walks.
\newblock {\em Lett.~Math.~Phys.}, 109(11):2571–2583, 2019.

\bibitem[Wad20]{Wada-2020}
K.~Wada.
\newblock A weak limit theorem for a class of long-range-type quantum walks in
  1d.
\newblock {\em Quantum~Inf.~Process.}, 19(1), 2020.

\end{thebibliography}

\newcommand{\etalchar}[1]{$^{#1}$}

\end{document}